\documentclass[format=acmsmall, review=false, nonacm=true]{acmart}
\usepackage{booktabs} 
\usepackage[ruled]{algorithm2e} 

\SetAlFnt{\small}
\SetAlCapFnt{\small}
\SetAlCapNameFnt{\small}
\SetAlCapHSkip{0pt}
\IncMargin{-\parindent}

\setcitestyle{authoryear}

\usepackage{bm}
\usepackage{bbm}

\usepackage{graphicx}
\usepackage{subfig}

\usepackage{algorithmic}

\usepackage{color}

\usepackage{xcolor}

\newcommand{\andy}[1]{#1}
\renewcommand{\vec}[1]{\mathbf{#1}} 

\theoremstyle{definition}

\newtheorem{theorem}{Theorem}[section]
\theoremstyle{remark}
\newtheorem*{remark}{Remark}
\newtheorem{corollary}{Corollary}[theorem]

\title{Auctioning Attention on Social Networks}

\author{Andy Lee}
\affiliation{%
  \institution{University of Illinois Urbana-Champaign}
  \city{Urbana}
  \state{IL}
  \country{USA}
}
\email{andy2@illinois.edu}

\author{Hari Sundaram}
\affiliation{%
  \institution{University of Illinois Urbana-Champaign}
  \city{Urbana}
  \state{IL}
  \country{USA}
}
\email{hs1@illinois.edu}

\begin{abstract}

Social media recommendation systems create conflict: content producers, content consumers, platform operators, and social pressures struggle to direct the allocation of attention in their favor while also facing competing societal pressures.
Producers may feel pressure to optimize for recommendation algorithms and consumers may be exposed to content with negative externalities such as polarization and misinformation. Platforms aim to maximize user engagement, oftentimes resulting in over consumption from consumers. There is mounting pressure from policymakers and broader society to address these issues.
Prior methods for constructing social media feeds have focused on recommendation systems. Instead, we propose a, to our knowledge, novel auction based method for feed construction where users bid for the attention of other users. Our mechanism systematically considers producers, consumers, platform operators, and social welfare.
We show that our auction is weakly incentive compatible in expectation under budget constraints with myopic agents and optimizes producer welfare. To balance between producer and consumer welfare, we introduce a tax policy to the auction to increase the cost of content with negative externalities.
Simulations over common social network topologies and an empirically observed network show how different feed algorithms prioritize different stakeholders. Our proposed auction based mechanisms produce on average 36.3\% higher producer welfare than comparison algorithms on the empirically observed network. Additionally, our methods produce 31.4\% higher producer welfare than comparison algorithms on the synthetic networks. Our methods also produce consistently more equitable distributions of attention than baseline methods across all evaluated network types.
Addressing all stakeholders is difficult; the incentives of each stakeholder are often mutually incompatible. Our proposed mechanism addresses attention allocation at a systematic level, balancing between the needs of different stakeholders.

\end{abstract}

\begin{document}


\maketitle



\section{Introduction}\label{sec:intro}

Stakeholders on social media platforms face fundamental tensions: recommendation systems seek to maximize user engagement, creating perverse incentives to produce content with negative externalities. In such an information economy, attention is a critical, scarce resource. As social media increases the quantity and speed of information spread, careful rationing of attention is more crucial than ever \cite{simon1996designing}. Attention is also the foundation of online advertising \cite{ahuja2025attention}; the efficiency of these systems depends on the ability to accurately model attention. Designing systems to effectively allocate attention on social media platforms is a key economic issue.
Content producers are strongly influenced by the mechanisms and algorithms chosen. How attention is allocated also has strong implications on the well being of users and those around them.
Interest in the quantity and quality of youth usage of social media has attracted policymaker attention \cite{Rosin_2025, bozzola2022use}.
Content consumers, content producers, social media platforms and broader society are all key stakeholders in the attention economy.

Despite the importance of attentional allocation, current interventions and systems do not address all stakeholders.
Content consumers are the focus of recommendation systems which target what is most captivating. Prior work has considered how to model attention rationing in individuals \cite{sims2003implications} but does not take a systematic view.
Content producers have focused their content on what performs best under recommendation systems \cite{bishop2018anxiety}. This has at times carried negative externalities such as incentivizing the production of misinformation \cite{diaz2025disinformation, pathak2023understanding} or toxic content \cite{regehr2025normalizing, kohl2024toxic}.
To combat this, platforms have deployed methods such as content moderation and downranking \cite{piccardi2024reranking} to modify the attention distribution away from content considered harmful.
Content moderation is labor intensive and users may feel that they are being censored \cite{cook2021commercial, myers2018censored}. While downranking may reduce the attention paid to harmful content, it does not impact the production of such content \cite{vincent2022measuring}.
Phone bans and screen time restrictions have also become increasingly popular methods for restricting access to social media \cite{bottger2024ban, Lambert_2025}.
Each of these methods addresses some stakeholder in the attention economy, but these piecemeal approaches create harmful incentives and negative externalities. Prior work does not systematically consider the competing incentives of all stakeholders.

In this work, we propose using auctions to construct user feeds in a setting where users are both content producers and consumers. To our knowledge, this is the first proposed system which uses an auction between users for feed construction.
Our mechanism takes users' private valuations for their posts as inputs for an autobidding \cite{aggarwal2019autobidding} system (Section \ref{sec:bidding}) to compute optimal bids for placement on followers' feeds.  An autobidder is used for practicality, users with large followings would not be able to manually bid on each follower's feed. Feed position is determined by a VCG auction.
We prove that our autobidding strategy is optimal for maximizing the expected number of impressions across the current and next iteration under budget constraints (Theorems \ref{thm:bellman_opt}, \ref{thm:myopic_bid}). This optimality implies that our mechanism is weakly incentive compatible under reasonable conditions (Theorems \ref{thm:truthful_bellman_bid}, \ref{thm:truthful_myopic_bid}), allowing us to measure the value of the attention that producers receive.
While the auction method maximizes producer welfare, we adopt a recommendation system style score for user interest as well as a score for post quality (for example misinformation content, toxicity) to measure consumer welfare. Consumer welfare considers a user's utility (for example, enjoyment) for consuming a post as well as any negative impacts such as misinformation. We suppose existing systems \cite{ko2022survey, islam2020deep, anjum2024hate} for computing consumer welfare are available as oracles. These oracles are used in a tax policy (Section \ref{sec:tax}) which imposes tighter budget constraints on users creating low consumer welfare content. This preserves the truthfulness property of our autobidder while reducing the spread of low consumer welfare content.
Our auction, combined with our tax policy, compose a mechanism which allows platforms to balance between stakeholder welfare with tunable parameters.
For general social concerns, our tax policy allows the platform to explicitly disincentivize content with negative externalities such as misinformation or toxicity. The stakeholders in our attentional allocation: content producers, consumers, the platform, as well as broader society, are all considered. Unlike prior methods such as content moderation and downranking which address only one stakeholder, our method allows for balancing between multiple stakeholders.
We summarize our contributions below:

\begin{description}
    \item[Auction based feeds] 
    We propose, to our knowledge, the first mechanism for constructing social media feeds via user to user auctions.
    Prior work focusing on producer welfare \cite{chen2025creator, yao2023rethinking, yao2024user} tends to focus on interventions which improve producer engagement within a recommendation system.
    By contrast, our mechanism uses an autobidder to compute the expected impression maximizing bid under budget constraints. Using the autobidder, we are able to maximize producer welfare. Through the optimality of the autobidder, we show that our mechanism is weakly incentive compatible. By the revelation principle, this allows us to accurately value the attention that is allocated. We show by simulation that the auction method also results in a substantially more equitable distribution of attention, differing from the commonly observed Pareto distribution of attention \cite{zhu2016attention}. Our methods consistently produce more equitable distributions of attention, including when degree distributions are highly concentrated.
    Our mechanism allows content producers to prioritize what is meaningful to them rather than what is optimal for a recommendation system. The more equal distribution of attention may also ameliorate some issues around incentives for producers to create content which contains misinformation, toxicity, and other negative externalities. Accurate attention valuation is also key to better understanding the priorities and off platform incentives of content producers. 
    \item[Policy for taxing attention]
    We introduce a tax policy in order to optimize consumer welfare.
    Prior work focuses on improving the performance and fairness of these oracles to increase consumer welfare \cite{li2024recent}.
    By contrast, we focus on using these oracles in a tax policy to increase the likelihood of higher consumer welfare posts being seen. Our tax policy allows platform operators to tune a weight for consumer and producer welfare, increasing the tax as consumer welfare is prioritized.
    We show by simulation how this tax policy may be used to effectively balance between consumer and producer welfare. Our tax policy limits negative externalities from content while still respecting producer welfare.
\end{description}
Our proposed auction methods produce on average 31.4\% higher producer welfare across three synthetic networks and produce 36.3\% higher producer welfare on an empirically observed network. Additionally, our methods produce more equitable distributions of attention than alternative feed algorithms. In our evaluations however, we find a tradeoff between producer welfare gain and attention distribution equity. We propose two variations of the auction mechanism with one producing higher producer welfare than the other. From our simulations, we find that the variation producing higher producer welfare typically produces less equitable distributions of attention.

As a reminder, we are interested in designing a mechanism which balances consumer and producer welfare while aligning with platform and social preferences. We leave as a policy question the issue of how young people engage with social media. The remainder of the paper is organized as follows: we first introduce an operationalization of the problem. Next we give theoretical results on the optimality of the mechanism in different settings. Finally, we evaluate the mechanism via simulation and demonstrate how an optimal tax policy may be learned via reinforcement learning.

\section{Related Work}\label{sec:rel}

\textbf{Misinformation detection and reduction:} Prior work exists for detecting misinformation across multiple modalities \cite{li2021entity, abdali2024multi} as well as on underlying social networks \cite{ozcelik2025detecting} .
Other work has focused on identifying groups producing coordinated misinformation \cite{zhang2021vigdet}.
There is also extensive prior work for toxicity detection \cite{sheth2022defining, taleb2022detection, garg2023handling}.
These prior works focus on quality scoring, not interventions for improving quality. We take these existing systems as oracles for computing a quality score.

Zareie and Sakellariou 2021 \cite{zareie2021minimizing} provides a survey of methods for reducing misinformation. The surveyed methods assume as we do that there is a pre-existing oracle for detecting misinformation. Strategies broadly either attempt to block information from spreading by severing edges from the network or attempt to seed counter information. Counter information methods are related to influence maximization problems \cite{li2018influence}, identifying the highest propagation vertices in a network.
In contrast, we do not introduce artificial content or reduce the total attention on the network by removing edges or users. These misinformation interventions also do not consider content producer welfare. Our tax policy limits the reach of misinformation without modifying the underlying network or allocating attention away from content producers.

\textbf{Content moderation: }A common method for handling toxicity is content moderation. Content moderation is more effective when personalized \cite{jhaver2023personalizing}, a highly labor intensive process.
Users broadly want platforms to moderate content rather than individually moderating their feeds \cite{cook2021commercial}. Users may also feel that content moderation is itself abusive or confusing, feeling that they are being censored or taken advantage of \cite{myers2018censored}.
Content labeling \cite{morrow2022emerging}, either by users or automated systems, may provide context for misinformation but will not reduce consumption or abuse.
Badges \cite{anderson2013steering} have also been considered as a method for encouraging desirable user behaviors however this may be ineffective against coordinated misinformation or abuse.
Content moderation methods do not consider content producer welfare. Our mechanism does not remove posts or users. Publishing low quality posts is costly under our tax policy. Varying the tax allows platform operators to balance between prioritizing producer and consumer welfare.

\textbf{Downranking: }Our proposed auction mechanism has similar effects to downranking though there are some crucial differences. Downranking is the practice of reducing a post's rank in a recommendation system so that it is less likely to appear in a user's feed \cite{piccardi2024reranking}. Similar to downranking, our intervention is at the feed level. While engagement on downranked posts is reduced, prior work finds that users are not discouraged from posting downranked content \cite{vincent2022measuring}.
Downranking ignores producer welfare, leading to the same perverse incentives as other recommender system based methods. Our auction mechanism reduces the spread of harmful content by reducing the budget of content producers who post low quality content while balancing producer welfare.

\textbf{Producer side interventions: }Prior work has also considered methods for encouraging the production of high quality content.
Models of content production as competitive games \cite{yao2023rethinking, yao2024user} allow for interventions which help producers reach equilibria maximizing consumer welfare.
These prior works take producer welfare as constant or given from platform rewards and thus may lead to the same incentives for producers to prioritize content with negative externalities.

The DualRec system \cite{chen2025creator} proposes improving producer welfare by creating a recommendation system which matches users to content instead of content to users. This method does not consider producer values; producers may still be incentivized to create content matching the largest number of users.
Prior work considering aggregate attention markets \cite{srinivasan2023paying} show that it is social welfare maximizing to show consumers a mix of high and low quality content.
In contrast, we consider individual feed construction, not aggregate markets.

\textbf{Autobidding: }We draw on prior work for developing our model of how users bid on the attention of others. In particular we base our model on statistical autobidding \cite{fernandez2017statistical} and find similar results to prior work on ad auction autobidding \cite{aggarwal2019autobidding}. As we will discuss, our model differs from prior work in that content producers only demand one slot on each targeted user's feed and have different objectives.
Substantial prior work also exists for auctions over social networks \cite{fang2024meta, guo2021emerging} however to our knowledge none consider auctions for constructing social media feeds.
Our use of reinforcement learning for tax policy discovery is inspired by prior work in using machine learning methods for mechanism design \cite{balcan2008reducing, dutting2019optimal, tang2017reinforcement} and tax policy discovery \cite{zheng2022ai}.

\section{Problem Description}\label{sec:problem}

\newcommand{\lgn}{l}
\newcommand{\ppr}{r}
\newcommand{\cmn}{C}
\newcommand{\qul}{q}
\newcommand{\atr}{\vec{A}}
\newcommand{\disc}{\beta}
\newcommand{\atc}{c}
\newcommand{\calib}{D}

\newcommand{\afn}[2]{I_{#1, #2}}
\newcommand{\pvl}[1]{v_{#1}}
\newcommand{\pus}{U}

\newcommand{\oAtn}{\kappa}
\newcommand{\oCwf}{\phi_C}
\newcommand{\oPwf}{\phi_P}
\newcommand{\oSwf}{\phi_S}
\newcommand{\oVwf}{\phi_V}

\newcommand{\psn}[2]{\pi_{#1, #2}} 

\newcommand{\qlm}{Q} 
\newcommand{\tgt}{\tau}
\newcommand{\tgtd}{\mathcal{T}} 
\newcommand{\atp}[2]{P_{#1, #2}} 
\newcommand{\atl}[2]{\alpha_{#1, #2}} 

We model a social network as a directed network $G=(V, E)$ over $T$ iterations. An edge $(i, j) \in E$ implies that user $i$ follows user $j$.
In each iteration users will simultaneously login, view posts on their feed, then make posts.
First we will describe the attributes of users and posts as well as how these attributes are used to measure the affinity that users have for posts.

\subsection{User affinities}\label{sec:affinities}

Users interests are based on community membership. We base interest on community membership in order to capture how our mechanism performs in settings where content producers must appeal to different communities.
Let $\mathcal{C}$ be the set of communities in $G$ s.t. $\mathcal{C}$ is a partition of $V$. We denote by $G^C = (V^C, E^C)$ the community graph, formed by merging all vertices within a community $C \in \mathcal{C}$. Each vertex $i \in V^C$ corresponds to some community in $\mathcal{C}$. For an edge $(i, j) \in E^C$ we have edge weight $\frac{\sum_{u \in i, v \in j} \mathbbm{1}[(u, v) \in E]}{\sum_{u' \in i, v' \in V \setminus i} \mathbbm{1}[(u', v') \in E]}$. From this graph $G^C$ we construct a max flow matrix $\bm{F^*}$ s.t. $F^*_{i, j}$ is the max flow value from community $i$ to community $j$ in $G^C$. When $j$ is not reachable from $i$ we assign some small $\varepsilon$ flow to $F^*_{i, j}$, similar to  PageRank \cite{ilprints422}.
An alternative to using the max flow matrix of $G^C$ would be to use a value such as expected hitting time. For our purposes, max flow is sufficient for creating vertices that produce content which is more or less interesting to other communities of vertices. Any alternative method over $G^C$ may be similarly substituted here.

Our users will have attributes $\atr_u \sim Dirichlet(\vec{F^*}_{C_u})$, drawn from the Dirichlet distribution parameterized by the max flows from their community $C_u$. We choose to model a user's interests this way as the max flow value approximates how close the interests of two communities are in the network. Similar to vertex similarity functions on undirected, unweighted graphs \cite{leicht2006vertex}, we are interested in the proportion of weighted paths between communities. The attribute of a post $p$ is given by $\atr_p \sim Dirichlet(\vec{F^*}_{C_{U_p}})$.
We measure how interested a user $u$ is in a post $p$ by $$\afn{u}{p} = \frac{\langle \atr_u, \atr_p\rangle}{\|\atr_u\| \|\atr_p\|}$$, the cosine similarity of the two vectors.
Next we discuss how we model the allocation of user attention.

\subsection{\andy{User attention}}\label{sec:attention}

The number of posts users view is based upon their interest in the posts on their feed and the marginal cost to allocating attention to the next post. We assume that users follow a hyperbolic discounting model to determine whether to view a post at rank $R+1$. The marginal cost to view a post at rank $R$ is denoted by $\atc(R)$. Users discount this cost by $\disc \in [0, 1]$; a perfectly rational long term thinking user would have $\beta = 1$. After viewing a post at rank $R$, $p_R$, the user decides if
\begin{equation}\label{eq:view_rule} 
    \afn{u}{p_R} - \disc \atc(R+1) \geq 0
\end{equation}
and if this inequality holds, the user views the post at rank $R+1$.
We suppose that $\atc$ is convex; this captures the marginally increasing costs of further investing attention into a feed. This cost may arise from due to opportunity cost or cognitive load. 
In this work we model $\atc$ as $\atc(R) = (\frac{R-1}{\calib-1})^2$ because it is simple and is clearly equal to $1$ when $R = \calib$. Here $\calib$ is a constant representing the expected depth at which a perfectly rational user would exit the feed. This constant is used to scale the expected number of posts seen for different network sizes. It is also used to model the scarcity of attention; large values of $\calib$ imply attention is less scarce.
Next we discuss how these components are used to construct the welfares of different stakeholders.

\subsection{\andy{Stakeholder welfare}}

We denote the quality of a post $p$ by $\qul_p$. Post quality here refers to social qualities such as toxicity or misinformation. We define the post quality distribution of user $u$ by $\mathcal{N}_{[0, 1]}(\qlm_u, \sigma)$ which is the normal distribution truncated to $[0, 1]$ having mean $\qlm_u$ and some standard deviation $\sigma$. We define post quality this way to model subsets of users who habitually generate low quality content, such as users participating in a coordinated misinformation campaign.
Users additionally have some private value $\pvl{p}$ for how much they value the attention received by their post $p$.

These components allow us to capture the competing welfare of four stakeholders. First, the total welfare of content producers is measured by $\oPwf(t) = \sum_{u \in V} \sum_{p \in \psn{u}{t}} \pvl{p}$. This represents the total private value of the attention given to content producers.
Next, the total welfare of content consumers is measured by $\oCwf(t) = \sum_{u \in V} \sum_{R \in [| \psn{u}{t} |]} \afn{u}{p_R} - \atc(R)$, the total interest of posts seen net the total attentional cost to view these posts.
Next, we consider the platform welfare which is given by $\oVwf(t) = \sum_{u \in V} | \psn{u}{t} |$, the total number of impressions.
Finally, we consider social welfare which is given by $\oSwf(t) = \sum_{u \in V} \sum_{p \in \psn{u}{t} } \qul_p$, the total quality score of all viewed posts.

These four notions of welfare shows the incentives of all four stakeholders. There is immediately competition between the incentives of the platform and the welfare of consumers; platforms are incentivized to drive consumers to overconsume past the threshold which would be rational. Similarly, there may be tension between producers and society as posters may highly value low quality posts. We will discuss how different types of feeds prioritize different stakeholders and how interventions may be used to balance between different stakeholders.


\begin{table}
\begin{center}
\caption{Problem definition notation}
\label{tbl:problem_notation}
\begin{tabular}{c c} 
 Notation & Definition \\
\toprule
$G = (V, E)$ & Network with users $V$ and edges $E$ \\
$T$ & Number of time steps \\
$\atr_u$ & Attribute vector of user $u$ \\
$\atr_p$ & Attribute vector of post $p$ \\
$\qul_p$ & Quality of post $p$ \\
$\pvl{p}$ & Value that $\pus_p$ has for post $p$ \\
$\afn{u}{p}$ & The interest that user $u$ has for post $p$ \\
$\psn{u}{t}$ & Posts seen by user $u$ in iteration $t$ \\
$\atc(R)$ & The marginal cost to view a post at rank $R$ \\
$\disc$ & The discounting users exhibit for the marginal cost to view the next post \\
$\calib$ & A constant used to scale the marginal viewing cost for different contexts \\
$\qlm_u$ & Average quality of a post made by user $u$ \\
$\oCwf(t)$ & The welfare of content consumers at $t$ \\
$\oPwf(t)$ & The welfare of content producers at $t$ \\
$\oVwf(t)$ & The welfare of the platform at $t$ \\
$\oSwf(t)$ & Social welfare at $t$ 
\end{tabular}
\end{center}
\end{table}


\subsection{Model}

In addition to the directly welfare relevant attributes, user $u \in V$ has a login probability $\lgn_u$, a post probability $\ppr_u$, and a community $\cmn_u$. We also denote the poster of a post $p$ by $\pus_p$.
We consider users that post and login at different rates in order to model observed behavior. Some users may not highly value the attention of others and will post infrequently, choosing primarily to consume content. A small group of users may dominate posting, generating the majority of the content on the network.

We consider a model in which all users act simultaneously in a given iteration. Within each iteration $t \in [T]$, each user $u \in V$ completes the following sequence of actions:
\begin{enumerate}
    \item $u$ draws a new attribute vector $\atr_u \sim Dirichlet(\vec{F^*}_{C_u})$ and with probability $\lgn$ logs in.
    \item If the user logs in, they view their feed. Posts are ordered in the feed of $u$ based on some feed algorithm (Section \ref{sec:feeds}). Each user $u$ view $\atl{u}{t}$ posts according to the user attention process described in Section \ref{sec:attention}.
    \item With probability $\ppr_u$, user $u$ makes a post.
    In this work we focus on posting only to the follower set of $u$, $N^-(u)$. While our model allows for sharing a post from $\psn{u}{t}$, we only simulate the setting where no post sharing occurs. We refer to the target set of a post by $\tgt_p$.
\end{enumerate}
We assume that each user has at most one active post at a time and posts only persist for one iteration. So a post made in iteration $t$ may only be seen in iteration $t+1$. We may denote the active post of user $u$ in iteration $t$ by $\atp{u}{t}$.

\begin{table}
\begin{center}
\caption{Model notation}
\label{tbl:model_notation}
\begin{tabular}{c c} 
 Notation & Definition \\
 \toprule
$\lgn_u$ & Login rate for user $u \in V$ \\
$\ppr_u$ & Probability that user $u$ posts \\
$\cmn_u$ & Community that user $u$ belongs to \\
$\pus_p$ & Poster of post $p$ \\
$\tgt_p$ & The target set of a post $p$ \\
$\atp{u}{t}$ & The active post of user $u$ in iteration $t$ \
\end{tabular}
\end{center}
\end{table}



\section{Overview of feeds}\label{sec:feeds}

\newcommand{\tps}[2]{\Psi_{#1, #2}} 
\newcommand{\RS}[2]{RS(#1, #2)} 

\newcommand{\ucr}[1]{\Delta_{#1}} 

Feed posts are ordered according to a feed algorithm. We consider three algorithms: a timeline algorithm where newer posts are placed first, an affinity algorithm where posts are ordered by consumer interest, and an auction mechanism where users bid for their posts to appear on the feeds of other users. In iteration $t$ we denote the posts that $u$ is eligible to see by $\tps{u}{t} = \{ p | u \in \tgt_p \land \exists\, v \in V s.t. \atp{v}{t} = p \}$. This is the set of active posts s.t. $u$ is a target of the post.

\subsection{Benchmarks}

We consider some benchmark feeds, each of which maximizes the welfare of a stakeholder.
The \textbf{affinity} feed is meant to model a recommendation system based feed. We assume that the goal of the recommendation system is to maximize user engagement. To this end, we suppose the recommender orders posts from highest to lowest affinity. We assume that the recommendation system observes the score of post $p$ for user $u$, $\afn{u}{p}$, directly. This gives the feed maximizing $\oVwf$.
The \textbf{consumer} feed maximizes consumer welfare, $\oCwf$. As in the affinity feed, posts are shown in order from highest affinity to lowest. The distinction is that when the net value of viewing the next post is negative, the consumer feed attempts to show the user a sufficiently low affinity post to cause the user to exit the feed. 
The \textbf{quality} feed maximizes our social welfare measure, $\oSwf$ the average quality of viewed posts. Posts are ordered from highest to lowest quality.
Finally, we consider a \textbf{timeline} feed which orders a feed $\tps{u}{t}$ in descending post iteration order. The newest posts are shown first with random tie breaks. This is used as a baseline feed which does not consider any form of welfare.

\subsection{Auction}\label{sec:feed_auction}

Our proposed mechanism (Section \ref{sec:bidding}) is a VCG auction equipped with an autobidder where users bid for positions on followers' feeds. We denote the number of credits a user $u$ has by $\ucr{u}$. 

For each post in a consumer's feed $p \in\tps{u}{t}$, we elicit a bid from the producer of the post $\pus_p$. We denote the bid of post $p$ to appear on the feed of $u$ by $b_{p, u}$. We base our bid $b_{p, u}$ on prior work in autobidding \cite{aggarwal2019autobidding, fernandez2017statistical} and discuss how the bid is computed in Section \ref{sec:bidding}.
After eliciting bids, we run a VCG auction over $\tps{u}{t}$. That is, for $\atl{u}{t}$ posts seen by $u$ in iteration $t$, bidders pay the $\atl{u}{t}+1$ highest bid. If all posts were seen, the posts pay $0$. We assume that users make at most one post in each iteration and so they bid on at most one slot per feed.

To balance between consumer and producer welfare, we introduce a tax policy which modifies $\ucr{u}$ based on the average consumer welfare of posts made by $u$ (Section \ref{sec:tax}). This limits the reach of users producing low quality posts.
In iteration $t$, our mechanism proceeds as follows:

\begin{enumerate}
    \item For a user making post $p$ in iteration $t-1$, they submit private value $\pvl{p}$ to the mechanism.
    \item The tax policy reallocates credits between users based on the average quality of their previous posts.
    \item The autobidder computes optimal bids based on the private values. We run a VCG auction on each feed using these bids. Posters pay the $\atl{u}{t}+1$ highest bid.
\end{enumerate}
The optimality of the autobidding strategy and the truthfulness of the bids $b_{p,u}$ are established in Section \ref{sec:bidding}.



\section{Bidding}\label{sec:bidding}

\newcommand{\blmndisc}{\delta}
\newcommand{\bic}{B}
\newcommand{\nexp}{\mathbb{E}[\pvl{p}^{t+1}]}
\newcommand{\marginal}{\theta}
\newcommand{\margfull}{\blmndisc \mathbb{E}[\bv'(\ucr{u}^{t+1})]}
\newcommand{\income}{\gamma}

When the budget constraint $\ucr{u}$ is binding, bidding one's true value is no longer an optimal strategy. The user will have to shade their bids to stay within their budget or limit the number of auctions they participate in. In the extreme case where $\ucr{u} = 0$, the user can only submit bids of $0$. 


\subsection{Myopic autobidding}

First, we will consider myopic agents who optimize over value weighted impressions in the current and next iterations. Next, we will relax this assumption of myopic agents to consider agents who optimize over the full time horizon.
This captures the opportunity cost of overspending in any given iteration. Though agents do not inherently value their credits, exhausting their credits reduces their ability to capture attention in the future. Our optimization is an extension of the statistical autobidding framework from prior work \cite{fernandez2017statistical}.

We first give the myopic form of the optimization program~\eqref{eq:myopic_autobid} for a post $p$ made in iteration $t$, and expand on the definitions afterward:

\begin{subequations}\label{eq:myopic_autobid}
\begin{align}
\max_{\vec{b}^t,\, \vec{b}^{t+1}}\quad & \sum_{u \in \tgt_p}  \pvl{p}\, w_u(b_u^t) + \nexp\,w_u(b_u^{t+1}) \nonumber \\
\text{subject to}\quad & \sum_{u \in \tgt_p} z_u(b_u^t) \leq \ucr{\pus{p}} \label{eq:bid-c1} \\
 & \sum_{u \in \tgt_p} z_u^{t+1}(b_u^{t+1}) \leq \ucr{\pus{p}} + \underbrace{\left (\income^t - \sum_{u \in \tgt_p} z_u(b_u^t) \right )}_{\text{Income - Costs}} \label{eq:bid-c2} \\
& \vec{b}^t \in [0, 1]^{|\tgt_p|} \label{eq:bid-b1} \\
& \vec{b}^{t+1} \in [0, 1]^{|\tgt_p|} \label{eq:bid-b2}
\end{align}
\end{subequations}

The cost $z_u(b^t)$ of bidding on the feed of $u$ given a bid $b_u$ is given by the function $w_u(b_u^t) \rho(w_u(b_u^t))$. The function $\rho_u(w)$ gives the expected payment for winning on the feed of $u$ given a win probability $w$ (we state our regularity assumptions on $w_u$ and $\rho_u$ as Assumption~2 below). We assume that our bids are bounded in $[0, 1]$, as enforced by constraints~\eqref{eq:bid-b1} and~\eqref{eq:bid-b2}. While we do not consider per user weights here, we note that they are straightforward to add and the analysis remains the same.


Our objective is the sum of the expected number of impressions in the current iteration and in the next iteration. The impressions are weighted by the realized private value of the post in iteration $t$ and by the expected value of the private value of the post in iteration $t + 1$.
Constraint~\eqref{eq:bid-c1} ensures that the expected cost of bidding $b_u^t$ does not exceed the credits the user has. The expected cost, $z_u$, is given by the probability of winning multiplied by the cost of winning.
Constraint~\eqref{eq:bid-c2} ensures that the expected cost of bidding $b_u^{t+1}$ does not exceed the expected number of credits the user will have in iteration $t+1$. The current credit amount $\ucr{\pus{p}}$, plus expected income $^t$, minus the cost of bidding $b_u^t$, gives the expected credits available in iteration $t+1$.
We make three assumptions in the myopic case, used throughout the analysis that follows.

\noindent\textbf{Assumption 1 (Myopic agents).} Agents select bids to maximize value-weighted impressions over the current and next iterations, $t$ and $t+1$, only.

\noindent\textbf{Assumption 2 $\nexp$ is known.} We assume that the platform knows the expected value of the private value of the next post the user will make.

\noindent\textbf{Assumption 3 (Regularity of $w_u$ and $\rho_u$).} For every user $u$, the win-probability function $w_u$ and the expected-payment function $\rho_u$ are once differentiable, increasing, and surjective.

By KKT analysis we show that the optimal strategy is to select two constant factors $\bic_1$ and $\bic_2$ such that $b^t_u = \frac{\pvl{p}}{\bic_1 + \bic_2}$ and $b^{t+1}_u = \frac{\nexp}{\bic_2}$.

\begin{theorem}[Optimal myopic bid existence]\label{thm:myopic_bid}
For post $p$, either bidding $1$ on all feeds is optimal or there are optimal bid constants $\bic_1$ and $\bic_2$ s.t. bidding $b^t_u = \frac{\pvl{p}}{\bic_1 + \bic_2}$ and $b^{t+1}_u = \frac{\nexp}{\bic_2}$ is optimal.
\end{theorem}

\begin{proof}
We solve this by demonstrating the KKT conditions on the Lagrangian. For constraints \ref{eq:bid-c1} and \ref{eq:bid-c2} we define the Lagrangian multipliers $\bic_1$ and $\bic_2$ respectively. For constraints \ref{eq:bid-b1} and \ref{eq:bid-b2} we define $\mu^+$ and $\mu^-$ for the upper and lower bounds of each $b^t_u$ respectively.
The Lagrangian is given by
\begin{align*}\label{eq:untax_myopic_lagrangian}
    \mathcal{L} &= \sum_{u \in \tgt_p} \pvl{p} w_u(b^t) + \nexp w_u(b^{t+1}) \\
    &- \bic_1 (\sum_{u \in \tgt_p} z^t(b_u^t) - \ucr{\pus{p}}) \\
    &- \bic_2 (\sum_{u \in \tgt_p} z^t(b_u^t) + z^{t+1}(b_u^{t+1}) - \ucr{\pus{p}} - \gamma) \\
    &- \mu^+ \sum_{u \in \tgt_p} (b^t_u - 1) + \mu^- \sum_{u \in \tgt_p} b^t_u
\end{align*}
First, we introduce some useful facts.
By definition $\rho_u(w_u(b)) = \frac{1}{w_u(b)} \int^{w_u(b)}_0 b(v) dv$, so that $\frac{\partial \rho_u(w_u(b))}{\partial b} = \frac{w_u'(b)}{w_u(b)} (b - \rho_u(w_u(b)))$. Accordingly then, $\frac{\partial z^t(b)}{\partial b} = w'(b) b$.

Next, we demonstrate the stationary values of $\bic_1$ and $\bic_2$.
Solving $\frac{\partial \mathcal{L}}{\partial b_u^t} = 0$ we have $\frac{\partial \mathcal{L}}{\partial b_u^t} = w'_u(b_u^t)(\pvl{p} - \bic_1 b_u^t - \bic_2 b_u^t) = 0$ when $b_u^t \in (0, 1)$. Simplifying, $b_u^t = \frac{\pvl{p}}{B_1 + B_2}$.
Similarly, solving for $\frac{\partial \mathcal{L}}{\partial b_u^{t+1}} = 0$ we derive $\frac{\partial \mathcal{L}}{\partial b_u^{t+1}} = w'_u(b_u^{t+1}) (\nexp - \bic_2 b_u^{t+1}) = 0$ and so $b_u^{t+1} = \frac{\nexp}{\bic_2}$ when $b_u^{t+1} \in (0, 1)$.
When $b_u^t = 0$, our stationary point is at $\pvl{p} w'(0) - \mu^+ + \mu^- = 0$. Similarly, when $b_u^t = 1$, our stationary point is at $w'(1)( \pvl{p} - \bic_1 - \bic_2) - \mu^+ + \mu^- = 0$.

Next we show that our candidate points satisfy feasibility.
Dual feasibility is given from the fact that $\bic_1, \bic_2, \mu^+, \mu^- \geq 0$.
Primal feasibility is given by the fact that costs are non-decreasing in the bids. By increasing $\bic_1$ and $\bic_2$, the bids decrease and so the costs decrease as well. We may satisfy primal feasibility if $\sum_{u \in \tgt_p} z_u(0) \leq \ucr{\pus{p}}$ and $\sum_{u \in \tgt_p} z_u(0) + z_u^{t+1}(0) \leq \ucr{\pus{p}} + \income$ which necessarily holds as a property of VCG.
Additionally, for any $\pvl{p}$ and $\nexp$ we may find $\bic_1$ and $\bic_2$ such that $\frac{\pvl{p}}{\bic_1 + \bic_2} \leq 1$ and $\frac{\nexp}{\bic_2} \leq 1$. From dual feasibility, it similarly follows that these bids are bounded from below by $0$.

Finally, we demonstrate complementary slackness.
Because $b_u^{t+1} = \frac{\nexp}{\bic_2}$, we must have that $\bic_2 > 0$. To satisfy complementary slackness then, $\sum_{u \in \tgt_p} z_u(b_u^t) + z_u^{t+1}(b_u^{t+1}) = \ucr{\pus{p}} + \income$ must hold. This exhaustion requirement is an intuitive conclusion for optimality; the objective is maximized when the second bid exhausts all remaining resources.
We consider complementary slackness by cases.

Case 1: $\pvl{p} > \nexp$

By our boundary conditions, $\pvl{p} \leq \bic_1 + \bic_2$ and $\nexp \leq \bic_2$. Thus, $\pvl{p} - \nexp \leq \bic_1$. If $\pvl{p} > \nexp$ then, $\bic_1 > 0$ and so constraint (1) must be active. Because $w_u$ and $\rho_u$ are continuous, surjective, increasing functions, the sum of their products $\sum_{u \in \tgt{p}} w_u(b_u^t) \rho_u(w_u(b_u^t))$ is a continuous, surjective, increasing function in the range $[0, \sum_{u \in \tgt_p} z_u(1)]$. When $\sum_{u \in \tgt_p} z_u(1) \geq \ucr{\pus{p}}$, we may compute $\bic_1 + \bic_2$ s.t. at the stationary point, constraint (1) is active. Similarly, we may find $\bic_2$ s.t. $\sum_{u \in \tgt{p}} w_u(b_u^{t+1}) \rho_u(w_u(b_u^{t+1})) = \income$ when $\sum_{u \in \tgt_p} z^{t+1}_u(1) \geq \income$. If either of these conditions fail, then the budget is non-binding and the user may optimally bid $1$ on all feeds.

Case 2: $\pvl{p} \leq \nexp$

If constraint (1) is active, then the analysis follows as before. Thus, suppose constraint (1) is inactive and so $\bic_1 = 0$. By the same argument as in case 1, $\sum_{u \in \tgt_p} z_u(b_u^t) + z^{t+1}_u(b_u^{t+1}) \in [0, 2 \sum_{u \in \tgt_p} \rho_u(1)]$.
If constraint (1) is inactive then $B_1 = 0$. Because constraint (2) is active, $\sum_{u \in \tgt{p}} w_u(b_u^t) \rho_u(w_u(b_u^t)) + w_u(b_u^{t+1}) \rho_u(w_u(b_u^{t+1})) = \ucr{\pus{p}} + \gamma$. Because $ + w_u(b_u^{t+1}) \rho_u(w_u(b_u^{t+1}))$ is a continuous, surjective, increasing function in the range $[0, 2 \sum_{u \in \tgt_p} \rho_u(1)]$. Thus if $\sum_{u \in \tgt_p} \rho_u(1) \geq \frac{\ucr{\pus{p}} + \gamma}{2}$ then there is a $b_u^t$ and corresponding $B_2$ s.t. constraint (2) is active. Otherwise, as before the agent is not budget constrained and may optimally bid $1$ for all targeted feeds.

\end{proof}

The argument is direct because the objective is linear in the win probabilities $w_u$ while the budget costs $w_u\, \rho_u(w_u)$ are convex; we derive closed forms for the full optimal bid vector in terms of the Lagrangian multipliers. This allows us to compute the optimal bids in polynomial time (Theorem~\ref{thm:myopic_bid_alg}) otherwise each bid for each one of the $|\tgt_p|$ feeds would be optimized individually.

\begin{theorem}[Optimal myopic bid computation]\label{thm:myopic_bid_alg}
The autobidder can compute the optimal bid for any post $p$ in polynomial time up to tolerance $\epsilon$.
\end{theorem}
\begin{proof}
We will use a combination of a root finding algorithm and the KKT conditions. Though our optimization problem solves for the current and next bids, in practice agents only need to solve for their current bid. We choose the \texttt{brentq} \cite{Brent1973} method though any bracketed root solving method may be used. We show by case analysis when solutions are in the interior and when they are at the boundaries. The first two cases cover the boundary conditions where the user is not budget constrained. Denote $\bic^+ = \max \{ \pvl{p}, \nexp \}$.

\textbf{Case 1: Constraint~\eqref{eq:bid-b2} is infeasible:}
When $$\sum_{u \in \tgt_p} z_u(\frac{\pvl{p}}{\bic^+}) + z^{t+1}(\frac{\nexp}{\bic^+}) < \ucr{\pus{p}} + \income^t$$, we assume that $\vec{b^t} = \vec{1}$ as the user is not budget constrained.

\textbf{Case 2: Constraint~\eqref{eq:bid-b1} is infeasible and $\pvl{p} > \nexp$:}
If $\pvl{p} > \nexp$, this implies that $\bic_1 > 0$. By complementary slackness, constraint~\eqref{eq:bid-c1} would also be active. If $\sum_{u \in \tgt_p} \rho_u(1) < \ucr{\pus{p}}$ then constraint~\eqref{eq:bid-c1} is not feasible as the user is not budget constrained. In this case as well, the user bids $\vec{b^t} = \vec{1}$.
 
\textbf{Case 3: Constraint~\eqref{eq:bid-c2} is feasible and $\pvl{p} \leq \nexp$:}
If ~\eqref{eq:bid-c1} is inactive then $\bic_1 = 0$ and we solve for $\bic_2 > 0$ such that ~\eqref{eq:bid-c2} is active. We may search for $\bic_2^* \in [0, \min \{ \pvl{p}^{-1}, \nexp^{-1} \}]$ such that Equation~\eqref{eq:case1-budget} holds:
\begin{equation}\label{eq:case1-budget}
\sum_{u \in \tgt_p} z_u(\frac{\pvl{p}}{\bic_2^*}) + z^{t+1}(\frac{\nexp}{\bic_2^*}) = \ucr{\pus{p}} 
\end{equation}
Finding such $\bic_2^*$ gives our bids $\vec{b^t}$ directly.

\textbf{Case 4: $\pvl{p} > \nexp$}
If $\pvl{p} > \nexp$ then $\bic_1 > 0$ and constraint~\eqref{eq:bid-c1} is active. If ~\eqref{eq:bid-c1} is feasible, we solve for $\bic_1 + \bic_2$ such that
\begin{equation}\label{eq:case2-budget}
\sum_{u \in \tgt_p} z_u(\frac{\pvl{p}}{\bic_1 + \bic_2}) = \ucr{\pus{p}}
\end{equation}
In this case we may search for $(\bic_1 + \bic_2)^{-1} \in [0, \pvl{p}^{-1}]$ such that Equation~\eqref{eq:case2-budget} holds.

The convergence rate using bisection search would be $O(-\log \epsilon)$ iterations for error $\epsilon$; generally \texttt{brentq} displays a faster convergence rate.
This optimization has a key implication for producer welfare optimization. Bids only directly consider post private values in case 3. In cases 1 and 2 the user is not budget constrained and submits the maximum bids. In case 4, the user exhausts their current budget, implying that whenever a post has a higher than average private value the user will submit a relatively large bid. Thus, depending on the state of the network, we should not expect all bids to be tightly correlated with private values.
\end{proof}

\begin{remark}\label{rem:single-period}
Dropping the next period term and constraint from program~\eqref{eq:myopic_autobid} leaves the fully myopic, single period objective $\max_{\vec{b}^t} \sum_{u \in \tgt_p} \hat{S}_{u, p}\, \pvl{p}\, z_u$ subject to $\sum_{u \in \tgt_p} z_u\, \rho_u(z_u) \leq \ucr{\pus{p}}$. As discussed above, when agents only seek to exhaust their current budget, bid computations do not consider the private value of the post. Thus, agents seeking only to maximize current iteration welfare will submit bids uncorrelated with their private value, reducing overall producer welfare. Looking ahead by at least one iteration requires users to weight the value of current impressions against those in the future.
\end{remark}

\subsection{\andy{Unbounded time horizon autobidding}}
\newcommand{\bv}{V}

Next we show how agents may relax from bidding myopically to bidding over an unbounded time horizon. In particular, we take the Bellman equation formulation of the optimization problem for user $u$ making post $p$.
We define the problem recursively as:
\begin{equation}\label{eq:bellman_bid}
   \bv(\ucr{u}) = \max_{\vec{b^t}} \sum_{u \in \tgt_p} \pvl{p}^t w_u(b^t) + \blmndisc \mathbb{E}[\bv(\ucr{u}^{t+1})]
\end{equation}
Subject to constraints
\begin{equation}\label{eq:bellman-c1}
     \sum_{u \in \tgt_p} z_u(b^t) \leq \ucr{u}
\end{equation}
\begin{equation}\label{eq:bellman-c2}
     \vec{b^t} \in [0, 1]^{|\tgt_p|}
\end{equation}
as before. $\blmndisc \in (0, 1]$ is a future discounting factor with smaller $\blmndisc$ implying greater impatience. Our recursive function $\bv$ represents the producer welfare gained from bidding optimal under given budget $\ucr{u}$. The analysis for the unbounded time horizon optimization is similar to that of the two step time horizon.

\begin{theorem}[Existence of optimal bids over an unbounded time horizon]\label{thm:bellman_opt}
For post $p$ made by user $u$, there is an optimal bidding strategy over an unbounded time horizon.
\end{theorem}
\begin{proof}
We begin by giving the first order conditions for the Bellman equation ~\eqref{eq:bellman_bid}. At the interior,
$$\mathcal{L} = \sum_{u \in \tgt_p} \pvl{p} w_u(b^t_u) + \blmndisc \mathbb{E}[\bv(\ucr{u}^{t+1})] - \bic (\sum_{u \in \tgt_p} z_u(b^t_u) - \ucr{u})$$ where $\bic$ is the Lagrangian multiplier as before.
As in the myopic setting, $z'^t_u(b^t_u) = w'_u(b^t_u)b^t_u$. Additionally, $\ucr{u}^{t+1} = \ucr{u}^t + \income^t - \sum_{u \in \tgt_p} z_u(b^t_u)$ and so $\frac{\partial \ucr{u}^{t+1}}{\partial b^t_u} = -z'^t_u(b^t_u)$.
Solving for stationarity, we get that
$\frac{\partial \mathcal{L}}{\partial b^t_u} = \pvl{p} w'_u(b^t_u) - z'^t_u (\margfull + \bic) = 0$.
Simplifying, $b^t_u = \frac{\pvl{p}}{\bic + \margfull}$ which mirrors our myopic optimal bid formulation.

By the envelope theorem, $\frac{\partial \bv}{\partial \ucr{u}} = \frac{\partial \mathcal{L}}{\partial \ucr{u}} = \margfull + \bic$ which is exactly the denominator for the stationary bid. In other words, our bids are a function of the private value of the current post $p$ and the marginal value of a credit at budget $\ucr{u}$.

Primal feasibility holds if $\sum_{u \in \tgt_p} z_u(b^t_u) \leq \ucr{u}$ when $b^t_u = \frac{\pvl{p}}{\bic + \margfull}$ at the interior. For sufficiently large $\bic$ this must hold as when $\bic \to 0$, $b^t_u \to 0$ and so $z_u(b^t_u) \to 0$ $\forall u \in \tgt_p$.
Dual feasibility holds as our optimal solution always lies in the region where $\bic \geq 0$; we examine this alongside complementary slackness.

If constraint ~\eqref{eq:bellman-c1} is active, then by complementary slackness, $\bic > 0$ and $\sum_{u \in \tgt_p} z_u(b^t_u) = \ucr{u}$. $z_u$ is non-decreasing in $b^t_u$ and so if ~\eqref{eq:bellman-c1} is active there must be some value of $b^t_u$ s.t. $\sum_{u \in \tgt_p} z_u(b^t_u) = \ucr{u}$ by the intermediate value theorem. If $\sum_{u \in \tgt_p} z_u(\frac{\pvl{p}}{\margfull}) > \ucr{u}$ then there must exist some $\bic > 0$ s.t. $\sum_{u \in \tgt_p} z_u(\frac{\pvl{p}}{\margfull + \bic}) = \ucr{u}$. If $\sum_{u \in \tgt_p} z_u(\frac{\pvl{p}}{\margfull}) \leq \ucr{u}$, then ~\eqref{eq:bellman-c1} is inactive and $\bic = 0$. Thus $\bic$ meets complementary slackness and dual feasibility.

Next we consider the boundary conditions. With the box constraints, the full Lagrangian multiplier is
$$\mathcal{L} = \sum_{u \in \tgt_p} \pvl{p} w_u(b^t_u) + \blmndisc \mathbb{E}[\bv(\ucr{u}^{t+1})] - \bic (\sum_{u \in \tgt_p} z_u(b^t_u) - \ucr{u}) - \mu^+ \sum_{u \in \tgt_p} (b^t_u - 1) + \mu^- \sum_{u \in \tgt_p} b^t_u$$ with additional Lagrangian multipliers $\mu^+$ and $\mu^-$ for the upper and lower boundaries respectively.
Taking $\frac{\partial \mathcal{L}}{\partial b^t_u} = 0$, we find $w'(b^t_u)(v - b^t_u(\bic + \margfull)) = \mu^+ - \mu^-$.
When $b^t_u = 1$, $\mu^+ = w'(1) (v - (\bic + \margfull)) + \mu^-$.
When $b^t_u = 0$, $mu^- = \mu^+ - w'(0) \pvl{p}$.
For primal feasibility, at $b^t_u = 1$, we have $\sum_{u \in \tgt_p} z^t(1) \leq \ucr{u}$ which holds when the user is not budget constrained. At $b^t_u = 0$ we have $\sum_{u \in \tgt_p} z^t(0) = 0$ which always holds.
Dual feasibility requires that if $b^t_u = 1$, $\mu^+ \geq 0 \iff v \geq \bic + \margfull$ and $w'(1) \geq 0$. If $v < \bic + \margfull$, the optimal solution must be at an interior point or at $b^t_u = 0$.
If $b^t_u = 0$, then $-w'(0) \pvl{p} > 0$ requires $w'(0) < 0$, however $w$ is non-decreasing. Thus $b^t_u = 0$ cannot be met, rather as $\ucr{u} \to 0$, $\bic \to \infty$ so that $b^t_u \to 0$.
Complementary slackness follows from being on the boundary.

\end{proof}

By Theorem \ref{thm:bellman_opt}, there is an optimal bidding strategy over an unbounded time horizon using the Bellman formulation. In particular, the bid is determined by the current private value, how binding the budget constraint is, and the marginal value of credits at the current budget. This formulation of the problem eliminates assumptions 1 and 2; the agents are not myopic and $\nexp$ is not assumed.
In contrast to autobidding for myopic agents however, agents optimizing over an unbounded time horizon must know the marginal value of their current budget. We do not assume that this is known, rather agents approximate this marginal value.

\begin{theorem}[Computing equilibrium optimal in expectation bids over an unbounded time horizon]\label{thm:bellman_alg}
Given $\marginal$, the marginal value of credits and user $u$ optimizing producer welfare over an unbounded time horizon, we may compute bids which are optimal in expectation at equilibrium for post $p$ of user $u$ in polynomial time.
\end{theorem}
\begin{proof}
In order to compute the optimal bid, the autobidder must compute $\bic + \margfull$. Given $\margfull$, $\bic$ may be computed as in the myopic search via a root finding method. If constraint ~\eqref{eq:bellman-c1} is infeasible, then we set $\bic = 0$ and the user bids $\min \{ \frac{\pvl{p}}{\margfull}, 1 \}$, the largest possible bid. Otherwise, we check two cases. If $\margfull = 0$, $\frac{\pvl{p}}{\margfull} > 1$, or $\sum_{u \in \tgt_p} z_u(\frac{\pvl{p}} {\margfull}) \geq \ucr{u}$ then constraint ~\eqref{eq:bellman-c1} must be active and we solve for $\bic > 0$ s.t. $\sum_{u \in \tgt_p} z_u(\frac{\pvl{p}}{\bic + \margfull}) = \ucr{u}$. If constraint ~\eqref{eq:bellman-c1} is not active, the user bids $\frac{\pvl{p}}{\margfull}$.

In order to compute this bid, the user must know $\margfull$. This is the marginal value of credits at budget level $\ucr{u}^{t+1}$ weighted by the future discount factor $\blmndisc$. Let us denote this by $\marginal^*(\ucr{u}) = \margfull$ In practice, this is difficult to estimate as the user only observes noisy results of different expenditures from their bids. Instead of learning $\marginal^*$ at all budget levels, agents learn a constant approximation for the marginal value, $\marginal$. We approximate $\marginal$ via the Robbins-Monro algorithm \cite{robbins1951stochastic}, taking $\marginal^{t+1} = \marginal - \frac{\ucr{u}^{t+1} - \ucr{u}^t}{t}$. This is approximating $\marginal$ such that $\mathbb{E}[\ucr{u}^{t+1} - \ucr{u}^t| \marginal] = 0$, the bid modification value at which the expected change in budget is zero. When bidding, we take $\margfull \approx \bar{\marginal}$, where $\bar{\marginal}$ is the running average of all computed $\marginal^t$ \cite{ruppert1988efficient}. Because at $\bar{\marginal}$ credit expenditure follows a Martingale, $\bar{\marginal}$ approximates $\margfull$ at the stationary distribution of $\ucr{u}$.

At equilibrium then, we may compute bids which are optimal in expectation. In general however, error can be introduced by estimating $\marginal$. If $\marginal$ is overestimated, bids are lower than optimal. No value of $\bic$ may recover the optimal bid as $\bic > 0$ and increasing $\bic$ may only decrease the bid. If $\marginal$ is underestimated, bids may be higher than optimal, though increasing $\bic$ may recover the optimal bid.
\end{proof}

Bidding over an unbounded time horizon substantially weakens the optimality guarantees attainable by the autobidder. Approximating the marginal value of each credit introduces a source of error. The approximation is the exact expected value of the marginal credits when the user's credits are at a stationary distribution. However, this approximation may result in over or underbidding and may be unstable initially.
An interesting consequence of our estimation of $\marginal$ is that the time discount, $\blmndisc$, is not explicitly accounted for. This implies that regardless of what time discounting a user experiences, the optimal bid strategy is the same.

\subsection{Weak incentive compatibility}\label{sec:truthful}

Suppose a user $u$ making bids over target set $\tgt$ with budget $\ucr{u}$ over two iterations. In the current iteration they have value $\pvl{p}^0$ and in the next, $\pvl{p}^1$.
Bidding is truthful in a standard VCG auction if bidding the true private value gives at least as high weighted impression value across both iterations than bidding anything else.
Formally, $\forall b \neq \pvl{p}^0$,
$\sum_{j \in \tgt} \pvl{p}^0 \mathbbm{1}[\pvl{p}^0 \geq P_j^0] + \pvl{p}^1 \mathbbm{1}[\ucr{u} - \sum_{l \in \tgt} \mathbbm{1}[\pvl{p}^0 \geq P_l^0] P_l^0 \geq P_j^1] \geq
\sum_{j \in \tgt} \pvl{p}^0 \mathbbm{1}[b \geq P_j^0] + \pvl{p}^1 \mathbbm{1}[\ucr{u} - \sum_{l \in \tgt} \mathbbm{1}[b \geq P_l^0] P_l^0 \geq P_j^1]$ where $P_j^t$ is the price to beat in iteration $t$ for the feed of user $j$.
In general, this may hold if the CDF of each targeted user is bilipschitz (upper and lower bound) with suitable constant, which would make truthful reporting optimal in expectation. The lognormal distribution, which many auctions follow \cite{pelto1971statistical, ballesteros2021standard}, has this property for some parameters. Because of our assumption that user values are normally distributed, the CDF of prices-to-beat per user feed tends to also follow a lognormal distribution. This property only holds for some CDFs however, and is unlikely to hold in general. Additionally, users may have very different CDFs as in a network with a power law degree distribution.

Though VCG is the underlying mechanism, users do not directly submit bids in our model. Instead, users report their value and the autobidder finds a bid that optimizes their payoff in this iteration and the next. As previously shown, for any given post value the autobidder finds the optimal bid for maximizing value across the two iterations. Providing the true value of a post to the platform is optimal. This does not imply that bids in our system are truthful. The optimal bid for a given true value may be larger or smaller than the given value. For instance when the budget is highly constraining the submitted bid may be much smaller and when the budget is completely non-binding the bid may be much larger. In Section \ref{sec:tax} we will discuss one potential method for selecting an appropriate starting budget.
\begin{theorem}[Weak incentive compatibility under myopic autobidding]\label{thm:truthful_myopic_bid}
The autobidder of Theorem \ref{thm:myopic_bid_alg} is weakly incentive compatible when the platform can accurately estimate $\nexp$ for all posts $p$ and users seek to maximize short term welfare.
\end{theorem}
\begin{proof}
Suppose a user $u$ provides value $\pvl{p}' \neq \pvl{p}$ for true value $\pvl{p}$. Let $b'$ be the current bid found for value $\pvl{p}'$ and $b$ be the corresponding bid for $\pvl{p}$.
Misreporting $\pvl{p}'$ is optimal if it provides higher utility across both iterations. However by the optimality of the autobidder, $b$ already gives the highest possible utility in expectation. Any other bid $b'$ provides at most equal utility; case reporting $v'$ is not an improvement. Reporting true value $\pvl{p}$ provides the highest expected payoff.
\end{proof}

\andy{
\begin{theorem}[Weak incentive compatibility under unbounded time horizon autobidding]\label{thm:truthful_bellman_bid}
The autobidder of Theorem \ref{thm:bellman_opt} is weakly incentive compatible when $\marginal$ is accurately approximated.
\end{theorem}
\begin{proof}
When $\marginal$ is correctly estimated for user $u$, by Theorem \ref{thm:bellman_opt}, we may derive the optimal bids under the Bellman equations. Misreporting the private value $\pvl{p}$ may not strictly improve achieved producer welfare.
The only case in which misreporting may strictly increase the producer welfare of $u$ is if $\marginal$ is misapproximated. If $\marginal$ is overestimated, $\pvl{p}$ may be correspondingly over reported in order to attain the true optimal bid. In essence, this recovers the bid that would be found by the autobidder under the true value of $\marginal$. Similarly, if $\marginal$ is underestimated, $\pvl{p}$ may be under reported.
\end{proof}
}
By the revelation principle, the platform treats reported values as truthful. This guarantee only holds in the myopic case and if $\nexp$ is known.
This result holds as long as we are able to accurately estimate $\nexp$. For example, if a user is planning to leave the platform, then $\pvl{p}^{t+1} = 0$ and misreporting is optimal. Here we restrict our analysis to the case where $\pvl{p}$ follows a known distribution and the network is static. Additionally, because our result of truthfulness only holds when bidders are budget constrained, initial budgets must be chosen carefully (Section \ref{sec:tax}).

\begin{corollary}
By the revelation principle, users directly provide their true values for each post under both autobidders. This achieves our goal of measuring the captured value on the platform.
\end{corollary}

These guarantees presume a suitable budget; bid values under the myopic autobidder are only correlated with private values under certain budget regimes. In Section \ref{sec:tax} we discuss how to appropriately set the credit level for content producers under the unbounded time horizon autobidder as well as the myopic autobidder.

\begin{table}
\begin{center}
\caption{Bidding notation}
\label{tbl:bid_notation}
\begin{tabular}{c c} 
 Notation & Definition \\
 \toprule
 $\ucr{u}^t$ & The credits available to user $u$ at iteration $t$ \\
 $w_u(b)$ & The probability of being viewed by user $u$ with bid $b$ \\
 $\rho_u(w)$ & The bid required to be viewed with probability $w$ by user $u$ \\
 $z_u(b)$ & The expected cost of bidding $b$ on the feed of $u$ \\
 $\bic$ & The Lagrangian bid constant used to compute the optimal bid \\
 $\income_u$ & The expected income for user $u$ \\
\end{tabular}
\end{center}
\end{table}

\section{\andy{Balancing between stakeholder welfares}}\label{sec:welfare}
Next we will discuss how different feed constructions impacts the welfare of each stakeholder. Consumer welfare and platform welfare conflict directly; platform welfare $\oVwf$ is highest when users over consume whereas consumer welfare $\oCwf$ is highest when user consumption is restricted. Similarly, producer welfare $\oPwf$ and social welfare $\oSwf$ may be directly opposed. Users participating in coordinated misinformation campaigns, for instance, may have low quality posts yet express extremely high values for impressions on their posts. It is not always desirable  to strictly maximize the welfare of all stakeholders.

Next we describe how tax policies may be applied to the auction feed in order to balance between producer welfare and social welfare. We focus on this tension specifically in order to model an intervention on a social media platform that must contend with coordinated misinformation or harassment campaigns. Thus, we focus on taxation based upon the quality score of posts. If we were instead interested in balancing between producer and consumer welfare, we could develop a tax policy based upon the affinity generated by posts. Balancing between all three stakeholders could be achieved by taking some function of the quality and affinity scores.

\subsection{\andy{Tax policies}}\label{sec:tax}

\newcommand{\wgt}{\omega}
\newcommand{\tax}{\eta}
\newcommand{\mbm}{M}
\newcommand{\mec}{W}
\newcommand{\txr}{R}

In Section \ref{sec:bidding} we discussed the importance of $\ucr{u}$ on the degree to which user's private values are factored into their bids. In this section we will discuss a tax policy implementation which allows a social planner to improve social welfare, $\oSwf$, and place users in a budget regime which allows private values to enter the autobidder. When budgets are too high, users are not budget constrained and always submit the maximum bid. Conversely, as budgets approach zero, bids approach zero independent of private value. To balance between different welfare measures, we must manage the walk that user budgets take.

\subsection{Budgeting under an unbounded time horizon}\label{sec:capped_tax}

We begin by discussing budgeting and taxation under an unbounded time horizon. In particular, how to manage the credit distribution for users under the Bellman autobidder.
First, we denote by $\mec_u = \sum_{v \in \tgt_u} z_v(1)$, the expected cost of submitting the maximum bid. Under the Bellman autobidder, when $\ucr{u} > \mec_u$, the user is not budget constrained.
Hence we set the user's initial budget, $\ucr{u}^0 = \mec_u$ as a heuristic. If $\ucr{u}$ follows a random walk with no drift, in expectation the user is exactly at the boundary of their budget constraint becoming active.

Next, we introduce a tax function, $\tax$. We denote by $\tax(u)$, the net taxes paid by or paid to user $u$. That is, $\ucr{u}^{t+1} = \ucr{u}^t + \eta(u)$. In order to distinguish between payments to and from the platform, we denote by $\tax^+(u)$ payments made to $u$ and we denote by $\tax^-(u)$ payments made by $u$. We require that $\tax$ preserve the \textbf{conservation of credits}, that is $\sum_{u \in V} \ucr{u}^t = \sum_{u \in V} \ucr{u}^{t+1}, \forall t$.

The first goal of $\tax$ to focus the distribution of $\ucr{u}$ around $\mec_u$. To this end, $\tax$ enforces a budget cap s.t. $\ucr{u} \leq \mbm \mec_u$ where $\mbm \in \mathbb{R}^+, \mbm \geq 1$. 
At $\mbm = 1$, $\eta$ forces $\ucr{u}^t = W_u, \forall t$. Increasing $\mbm$ allows for a larger budget differential between users. By conservation, the pool of credits is fixed and so a larger cap allows for greater concentration of credits among fewer users.

The second goal is for $\tax$ to redistribute credits based on the quality of posts produced by users. This is in order to increase $\oSwf$. Generally, we would like for $\tax$ to produce a distribution of budgets where users generating low social welfare have reduced budgets and reach. By conservation this implies that users producing high quality posts are given greater reach.

We define a tax rate, $\txr \in [0, 1]$ which allows for a social planner to control the degree of redistribution from $\tax$.
The taxes charged are given by:
\begin{equation}\label{eq:charge_tax}
    \tax^-(u) = \max(\ucr{u}^t - \mbm \mec_u, 0) + \txr (1 - \qul_{\atp{u}{t}})^2 \min(\ucr{u}^t, \mec_u)
\end{equation}
which caps the budget in the first term and charges a quality based tax in the second term.
The tax distribution takes $\sum_u \tax^-(u)$ and redistributes this proportionally by quality. In order to respect budget caps, we must construct $\tax^-$ iteratively. First, we must determine if any users would exceed or meet their budget cap after their allocation net their charged tax. If so, we allocate these users a payment $\tax^+(u) = \mbm \mec_u - (\ucr{u} - \tax^+(u))$. This value is non-negative as, by construction, $\ucr{u} - \tax^+(u) \leq \mbm \mec_u$. If there are no such users, we allocate 
\begin{equation}\label{eq:distr_tax}
    \tax^+(u) = \frac{\qul_{\atp{u}{t}}}{\sum_{v \in U} \qul_{\atp{v}{t}}} \sum_{v \in U} \tax^-(v)
\end{equation}
where $U$ is recursively defined as $U = \{ v | v \in V, \ucr{v}^t - \tax^-(v) + \tax^+(v) \leq \mbm \mec_v \}$, the set of users who are below their budget capacity. This distributes the pool of charged taxes proportionally with the relative quality of a user's post. For clarity, we illustrate this procedure in pseudocode as well (Algorithm \ref{alg:distr_taxes}).

\begin{algorithm}
\caption{Tax distribution, $\tax^+$}\label{alg:distr_taxes}
\begin{algorithmic}
\STATE $U\gets V$
\STATE $D\gets\sum_{v \in V} \eta^-(v)$
\WHILE{$|U| > 0$}
    \STATE $\hat{\qul_u} \gets \frac{\qul_{\atp{u}{t}}}{\sum_{v \in U} \qul_{\atp{v}{t}}} \forall u \in U$
    \STATE $X\gets\{ v | v \in U, \ucr{v}^t - \tax^-(v) + \hat{\qul_v} D > \mbm \mec_v \}$
    \IF{$X \neq \emptyset$}{
        \STATE $\tax^+(v) \gets \mbm \mec_v - \ucr{v}^t + \eta^-(v) \forall v \in X$
        \STATE $D \gets D - \sum_{v \in X} \eta^+(v)$
        \STATE $U \gets U \setminus X$
    }
    \ELSE
        \STATE $\tax^+(v) \gets \hat{\qul_v} D \forall v \in U$
        \STATE $U \gets \emptyset$
    \ENDIF
\ENDWHILE
\end{algorithmic}
\end{algorithm}

This tax function $\tax$ gives a social planner or platform two tunable parameters: $\txr$ and $\mbm$. Setting $\txr$ determines the degree to which the platform is willing to sacrifice producer welfare for social welfare. As $\txr$ approaches $1$, users producing low quality post will have extremely restricted reach even if they have high private values for the content. Indeed in the case where a user is participating in coordinated misinformation this may even be likely; such users will have uniformly high private values for all posts. Similarly, increasing $\mbm$ determines the degree to which a tax can concentrate credits. At $\mbm = 1$ even a high tax rate of $\txr = 1$ does not shift the distribution of taxes as conservation prevents budgets from deviating from the initial values of $\mec$.

\subsection{Myopic budgeting}\label{sec:power_tax}

The myopic budget setting relies upon Assumption 2 of the myopic autobidder. Given access to the expected value of the next post, $\nexp$, we set $\ucr{u}^0 = \nexp \cdot | \tgt_u |$ where $\tgt_u$ is the target set of user $u$. In this case, $\tgt_u$ is the follower set of $u$. In the myopic setting we do not take a budget cap. Rather, we tax the quality directly.
Thus, the tax charge is given by
\begin{equation}\label{eq:myopic_tax_charge}
    \tax^-(u) = \ucr{u}^t (1 - \qul_{\atp{u}{t}})^\txr
\end{equation}
where $\txr \in \mathbb{R}^+$ is the tax rate. This charges each user a proportion of their budget based upon the quality of their post. A high quality post is charged very little, a low quality post is charged a lot.
The tax redistribution is given by
\begin{equation}\label{eq:myopic_tax_distr}
    \tax^+(u) = \ucr{u}^t \frac{\qul_{\atp{u}{t}}}{\sum_v \qul_{\atp{v}{t}}} \sum_v \tax^-(u)
\end{equation}
which reallocates the pool of charged credits based upon the relative quality of each users post. A user with high relative quality receives a larger share, and a user with low relative quality receives a smaller share.

This tax only bounds the credits of each user to be non-negative. Conservation of credits is still preserved as in the unbounded case. We must adopt this form of initial budget and tax because the distribution of budgets wherein a user is not budget constrained is smaller in the myopic case than in the unbounded case. This forces the budget to be correlated with the expected private values of the user. Capping budgets may also sufficiently perturb users budgets such that they are no longer in the private value correlated budget regime. This brittleness also leads to a weaker performance of the myopic autobidder as compared to the unbounded time horizon autobidder.

\newcommand{\qem}{\mathbb{E}[\qul_p^\txr]}

The inclusion of a non-zero tax also requires the myopic autobidder to consider the tax in order to estimate the next iteration budget. The program is largely unchanged, however constraint ~\ref{eq:bid-c2} becomes
\begin{equation} \label{eq:taxbid-c2}
\sum_{u \in \tgt_p} z_u^{t+1}(b_u^{t+1}) \leq \ucr{\pus{p}} + \income^t - \sum_{u \in \tgt_p} z_u(b_u^t) + \mathbb{E}[\tax^+] - \mathbb{E}[\tax^-(\vec{b})]
\end{equation}
where $\mathbb{E}[\tax^-(\vec{b})] = (\ucr{\pus{p}} + \income^t - \sum_{u \in \tgt_p} z_u(b_u^t))(1 - \qem)$. We estimate $\mathbb{E}[\tax^+]$ as a constant; though the value depends on the current bid vector, we assume a large network in which the marginal effect of $\vec{b}$ is negligible. Additionally, as we assumed that $\nexp$ was known, now we must additionally assume $\mathbb{E}[\qul_p]$ is known.
Rearranging some terms, constraint (2) is equivalently $\sum_{u \in \tgt_p} \qem z_u(b^t) + z_u(b^{t+1}) \leq \qem (\ucr{\pus{p}} + \income) +  \mathbb{E}[\tax^+]$.
By KKT analysis, the new bid is $b^t = \frac{\pvl{p}}{\bic_1 + \bic_2 \qem}$ where $\qem \in [0, 1]$. The stationary value of $b^{t+1}$ is unchanged. Because this does not change the domain of achievable bids, introducing a tax only changes the regimes in which different constraints are active. Feasibility and complementary slackness are similarly unaffected.

\section{Empirical evaluation}\label{sec:eval}

We evaluate the performance of our model empirically via simulation. We run simulations over three stylized random graphs in order to isolate the effect of different topologies on welfare outcomes. Additionally, we set up the parameters of our model to maximally demonstrate the difference in welfare generated by different feeds. Finally, we demonstrate the outcomes of different feeds on a more empirically grounded network with empirically derived parameters to estimate the effect of different feeds on a real social network.


\subsection{Simulation setup}

In order to understand the relative advantages of different feed ordering algorithms, we will consider four types of feeds. We consider the consumer, platform, and social welfare maximizing feeds as baselines to illustrate upper bounds on the respective stakeholder welfare. Then, we consider our proposed Bellman, unbounded time horizon auction feed in which agents are not myopic and consider the opportunity cost of future impressions. We do not report the stakeholder welfare produced by the myopic autobidder based auction feed as it is strictly dominated in all settings by the Bellman autobidder.

To provide an intuition on the performance of the different feeds as well as the effect of network topology on the different feeds, we begin with a simplified set of parameters.
For a network $G = (V, E)$, $\forall u \in V$ we set $\lgn_u = 1$ and $\ppr_u = 1$, users always login and post. User $u$ draws private values for their posts $p$ by $\pvl{p} \sim Normal_{[0, 1]}(\mu_u, 0.15)$, the Normal distribution truncated to $[0, 1]$ where the mean $\mu_u \sim Uniform(0, 1)$. We spread the mean in order to more clearly demonstrate each feeds ability to increase producer welfare. We also split the users into two types, half of the users make posts of low quality and half of the users make posts of high quality. Low quality posters have $\qul_p \sim Normal_{[0, 1]}(0.3, 0.1)$ and high quality posters have $\qul_p \sim Normal_{[0, 1]}(0.7, 0.1)$. We split the users in this way to better highlight the difference in performance between the feed algorithms.

In Section \ref{sec:bidding} we demonstrated optimal bidding strategies for the myopic and Bellman auctions. However as stated in Assumption 3, these strategies require that we have access to the functions $w_u, \rho_u \forall u \in V$. In order to estimate these functions, we first run a simulation of the VCG auction with users bidding randomly. We set $b_u = \pvl{p}$ for each post $p$. From these random auctions we are able to estimate an empirical CDF over the winning bids for each user. For space and computationally efficiency we do not save all bids, instead we compute a histogram. Because the winning bid values tend to be clustered, we do not space the bins of the histogram uniformly. Instead, we create bins with cutoffs exponentially spaced by $\frac{e^{10x} - 1}{e^{10}-1}$ for $x$ linearly spaced in $[0, 1]$. Choosing an appropriate bin count is essential, this impacts the resolution of the estimated CDF and the accuracy of the estimated $w$ and $\rho$ functions.
We choose the number of simulation iterations and the bin count by the DKW inequality \cite{massart1990tight, dvoretzky1956asymptotic} with an error of $0.05$ and a confidence of $0.9$.
Because we use VCG, the produced winning bid CDF is lognormal though we use an empirical CDF to demonstrate an operationalization that is useful for more general distributions.
In the myopic setting, we additionally assume that private values and post qualities have known means.

\subsection{Evaluation networks}

We first consider three network structures, each of which highlight a particular network topology. These structures will also help define our user attribute vectors.
The first structure is the $k$-regular ring lattice \cite{watts1998collective}. In this network, vertices are placed on a ring and are adjacent to the $\frac{k}{2}$ nearest neighbors. Because this exhibits no community structures, all users are in the same community. Thus in the $k$-regular lattice, all users are equally interested in all posts. This topology will provide the simplest example of how the different feed types prioritize different welfare objectives.

The second structure is the stochastic block model \cite{holland1983stochastic}. The standard stochastic block model partitions the network into $k$ blocks of size $m$ each s.t. $|V| = km$. The probability of an edge existing between two vertices of the same block is given by $p$ and the probability of an edge existing between two vertices of different blocks is given by $q$ where $p > q$. This block structure gives a natural partition of each block as its own community; our attribute vectors are derived from this block structure.

The third structure is a scale free network \cite{bollobas2003directed}. This network exhibits a power law degree distribution similar to that of a Barab\`asi-Alberts network \cite{albert2002statistical}, however the scale free network is directed. These networks again have no inherent community structure so we again ignore this to isolate the effect of the degree distribution. The power law degree distribution produces significantly different demands on the attention of different vertices. Bidding on very competitive vertices becomes very expensive. The scale free network demonstrates performance when vertices have extremely heterogeneous degrees.

We also wish to evaluate our mechanism on an empirical social network. Our network dataset inclusion criteria are: directed edges, ground truth labeled communities, edges represent interaction. Because of our modeling assumptions, we also require that each vertex belong to exactly one community.
Many social networks do not have ground truth communities, have structural sampling bias \cite{leskovec2012learning, gupta2015structural}, or are undirected. Additionally, commonly scraped social networks do not accurately capture attentional allocation. For instance, two users may be friends even if they never interact. One could use datasets directly capturing user interactions on social media, however these may be influenced by recommendation systems which complicate network structure. Finally, many networks with labeled communities are multigraphs which we do not consider.
One network meeting our criteria is the EU email network \cite{leskovec2007graph} ($|V| = 1005, |E| = 25571$). In addition to meeting our inclusion criteria, the EU email network exhibits topological features common to other social networks such as small diameter and skewed degree distribution (Figure \ref{fig:eu_deg}, Appendix). This network represents communication between individuals, better capturing our notion of attentional demand than networks with inactive edges.
We set login and post rates on the EU email network to match the distributions from Pew survey data \cite{mcclain2021behaviors}.

For the $k$-regular lattice, SBM, and scale free networks, we consider $|V| = 100$ and an average degree of $6$. We fix the rational attention limit of all vertices at $3$. We choose this average degree and rational attention limit so that all feeds have more posts available than available attention. Each simulation is run for $200$ iterations.

\subsection{\andy{Welfare results}}

We include our simulation parameters in Table \ref{tbl:sim_params} below.
\begin{table}
\begin{center}
\caption{Simulation parameter values}
\label{tbl:sim_params}
\begin{tabular}{c c c} 
 Parameter & Definition & Value \\
\toprule
$|V|$ & Network size & $100$ \\
Bellman $\txr$ & Tax rates & $\{ 0.0, 0.5, 1.0 \}$ \\
Myopic $\txr$ & Tax rates & $\{ 0.0, 5.0, 10.0 \}$ \\
$\mathbb{E}[N^-(u)]$ & The average in-degree (posts available) & $6$ \\
$\calib$ & Rational attention limit (max posts seen) & $3$ \\
$\mbm$ & Budget cap at $\mbm$ times their maximum per iteration expenditure & $2$ \\
$T$ & Total iteration count per simulation & 200 \\
\end{tabular}
\end{center}
\end{table}

In order to compare feeds which may have different view counts, we normalize $\oCwf, \oPwf, \oSwf$ by view count. We show results with $\txr \in \{ 0.0, 0.5, 1.0\}$ to illustrate results in a variety of $\txr$ levels in the unbounded time horizon autobidder. In the myopic setting, we take $\txr \in \{ 0.0, 5.0, 10.0\}$. In our evaluations we set $|V| = 100$. We target an average in-degree of $6$ across all three synthetic networks. For all three networks we set the rational attentional limit as a constant $\calib = 3$. This forces all feeds to make tradeoffs in choosing posts. Each simulation is run $25$ times.

We find that producer welfare is optimized in all three synthetic networks by our auction feeds (Figure \ref{fig:wef_prod}). This is because our auction feeds are the only feeds which elicit content producer's private values for their posts. Consumer welfare is maximized by the consumer feed (Figure \ref{fig:wef_consumer}). At low values of $\beta$, the hyperbolic discount factor on the marginal cost of viewing posts, almost all posts are viewed regardless of feed. This results in the difference in welfare shrinking between feeds. Similarly, at $\beta=1$, exactly $\calib$ posts are viewed and so there is no difference in consumer welfare. At moderate values of $\beta$, however, the consumer feed successfully causes users to terminate their sessions at the over consumption minimizing rank. This results in the consumer feed generating the highest consumer welfare. Social welfare is maximized by the quality feed (Figure \ref{fig:wef_social}). The inclusion of a high tax rate has modest effects in increasing the average viewed quality. The redistributive effect of the tax is hampered by the budget cap imposed on each user. At our chosen value of $\mbm=2$, users submitting high quality posts hit their budget caps and the excess credits are redistributed to users submitting low quality posts.
Finally, platform welfare (Figure \ref{fig:wef_platform}) is maximized by the affinity feed. This feed orders posts from most to least engaging, encouraging over consumption from users. The effect is visible in the stochastic block model, the only synthetic network where users have heterogeneous post affinities. At moderate values of hyperbolic discounting, feeds which do not explicitly optimize for user engagement may cause users to lose interest and exit the feed earlier. At extreme values of hyperbolic discounting, feed ordering has no effect on platform welfare as users either view all posts in their feed or view exactly $\calib$ posts.

\begin{figure}%
    \centering
    { \includegraphics[width=0.98\linewidth]{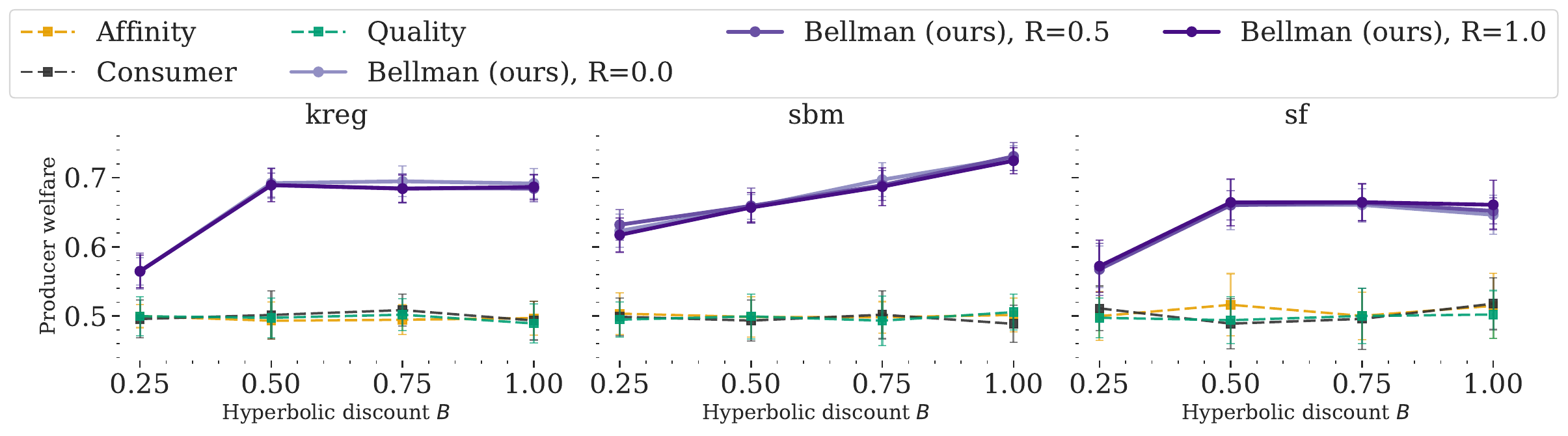} }%
    \caption{Producer welfare, $\oPwf$, by feed. Producer welfare is given by the average private value for viewed posts. Our proposed method, the auction, is the only feed able to optimize for producer welfare as it is the only feed which elicits users' private values for their posts.}%
    \label{fig:wef_prod}%
\end{figure}

\begin{figure}%
    \centering
    { \includegraphics[width=0.98\linewidth]{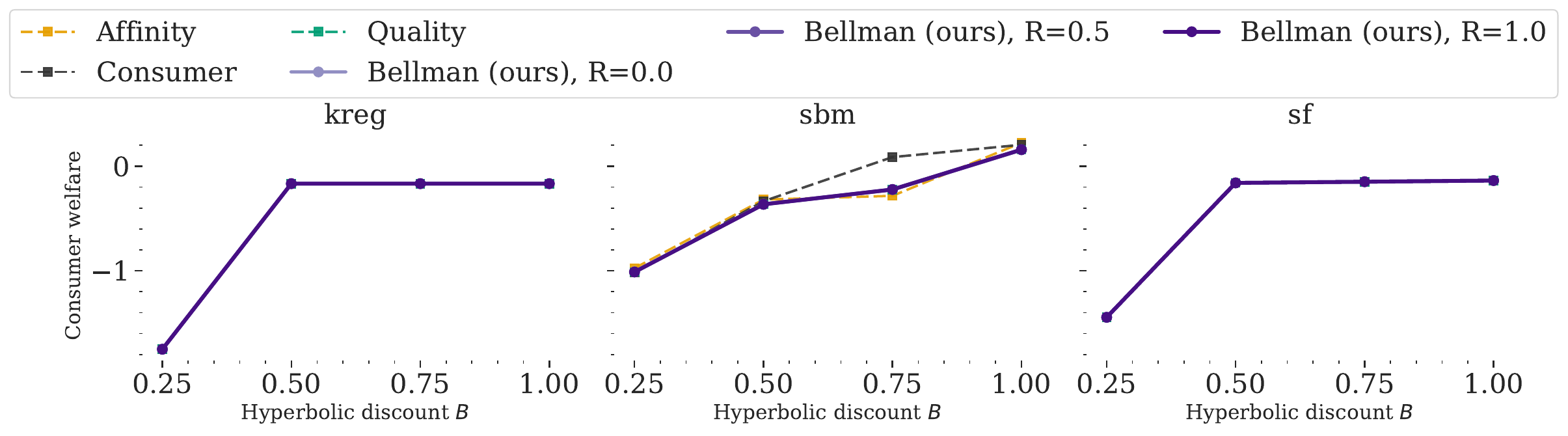} }%
    \caption{Consumer welfare, $\oCwf$, by feed. Consumer welfare is dependent on users' affinities for the posts they view as well as the degree to which they over consume content. Our consumer welfare feed limits overconsumption, increasing consumer welfare. By contrast, the affinity feed encourages over consumption, decreasing consumer welfare.}%
    \label{fig:wef_consumer}%
\end{figure}

\begin{figure}%
    \centering
    { \includegraphics[width=0.98\linewidth]{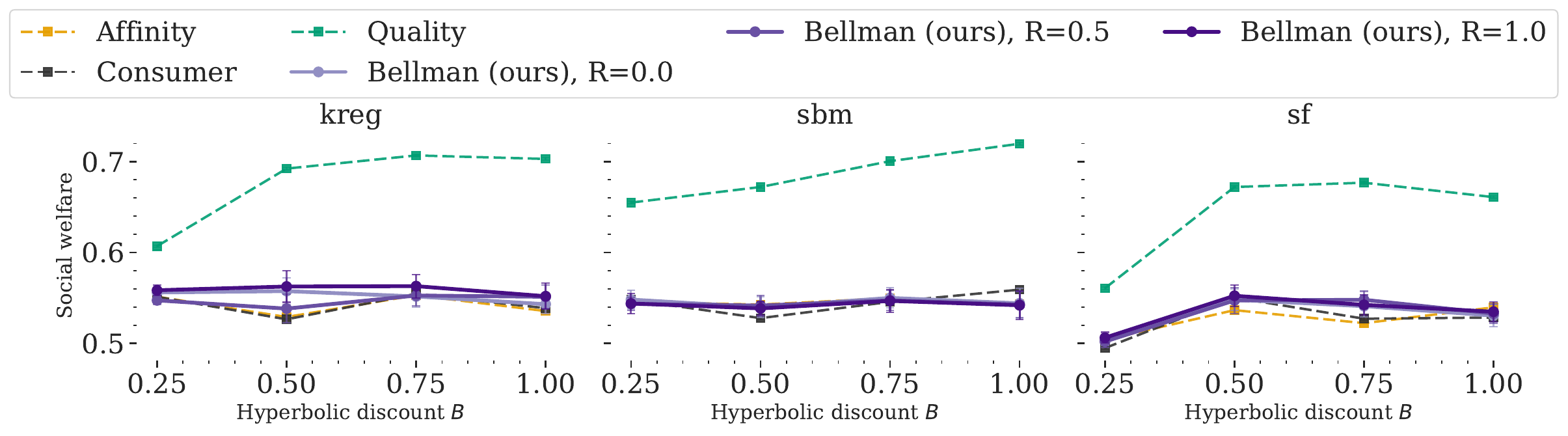} }%
    \caption{Social welfare, $\oSwf$, by feed. Our social welfare measure is the average quality of viewed posts. The quality sorted feed is optimal in this case as it sorts posts directly from highest quality to lowest.}%
    \label{fig:wef_social}%
\end{figure}

\begin{figure}%
    \centering
    { \includegraphics[width=0.98\linewidth]{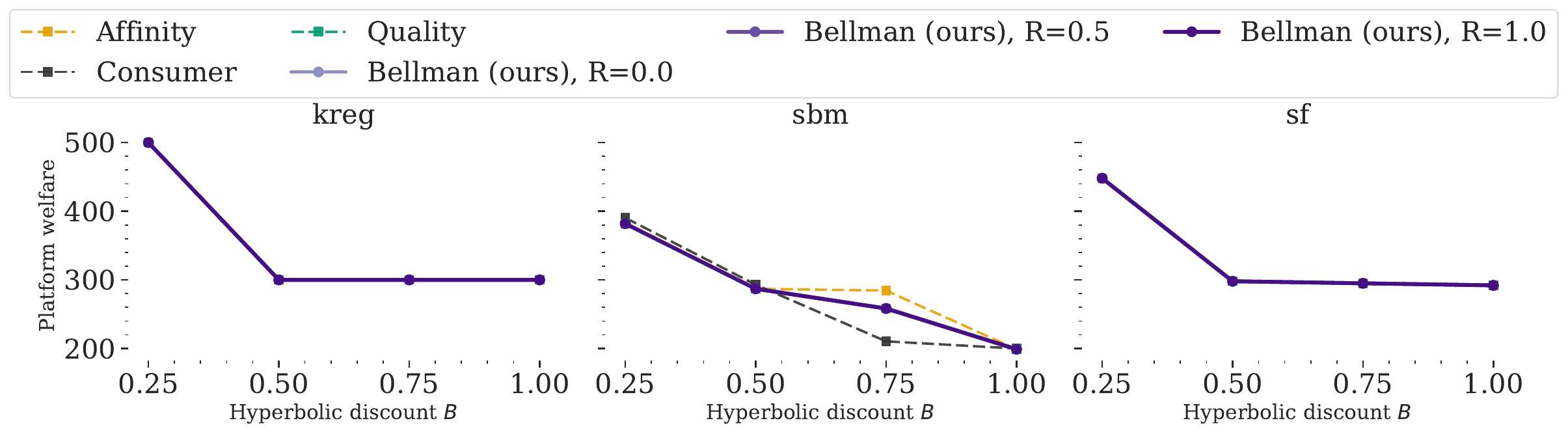} }%
    \caption{Platform welfare, $\oVwf$, by feed. Our platform welfare measure is the total attention captured by the platform, the number of views. The affinity feed, which optimizes for platform welfare by strictly ranking posts from highest to lowest affinity, maximizes posts seen. This can be seen in the SBM network.}%
    \label{fig:wef_platform}%
\end{figure}

The patterns in stakeholder welfare we observe over synthetic networks largely hold in our empirical network as well (Figure \ref{fig:wef_eu}). Producer welfare is maximized by our auction feeds with higher tax rates generating higher producer welfare. This is because more redistribution keeps producers from depleting their credits, maintaining the correlation between private value and bid value. Consumer welfare is maximized by the consumer feed; the consumer feed is the only feed which prevents over consumption at moderate levels of hyperbolic discounting. Social welfare is maximized by quality ordering however the auction feed with a high tax rate also improves upon social welfare. Finally, platform welfare is maximized by the affinity ordered feed. This feed maximizes total engagement at the cost of consumer over consumption. On the EU network, the auction feed produces on average 36.3\% more producer welfare than the baseline affinity ordered, quality ordered, and consumer welfare maximizing feeds.

\begin{figure}%
    \centering
    { \includegraphics[width=0.98\linewidth]{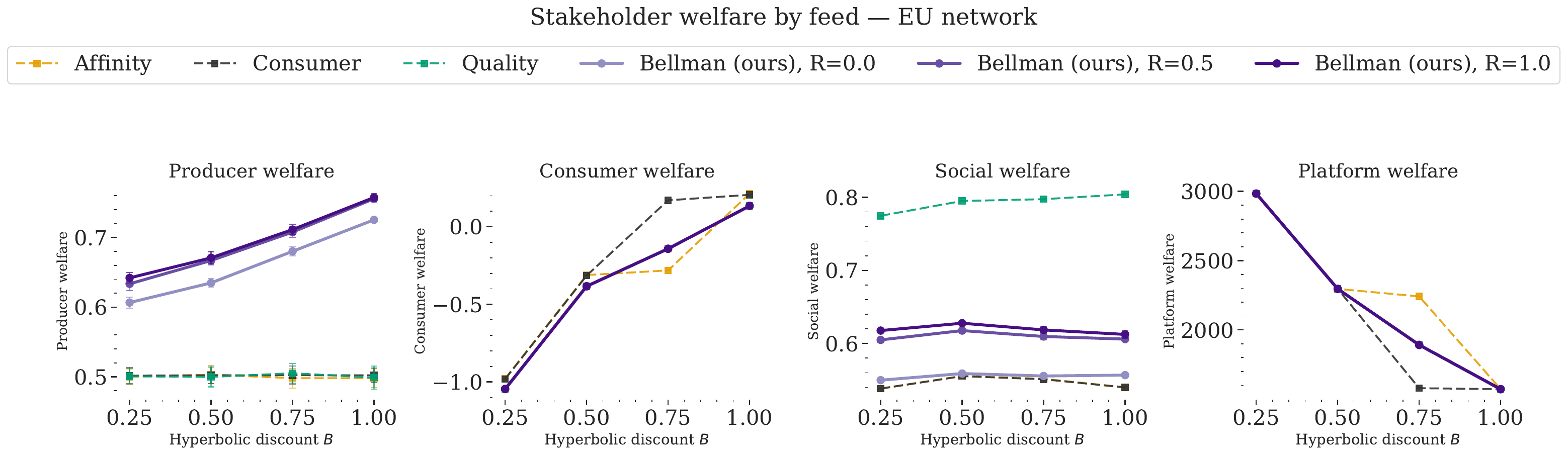} }%
    \caption{Stakeholder welfare by feed on the EU core email network. Consumer and platform welfare are only differentiated by the consumer and affinity feeds. Increasing the tax rate increases the social welfare generated by the auction feeds however social welfare is maximized by quality based ordering.}%
    \label{fig:wef_eu}%
\end{figure}

\subsection{Ablations}\label{sec:ablations}

We evaluate the sensitivity of our proposed mechanisms to the violation of their necessary assumptions. This evaluation is done via ablation simulations; we systematically vary the degree to which assumptions are violated and observe the impact on agent behavior and welfare.
For our ablations, we study the effect of misestimating parameters on producer welfare by fixing one agent at a time, ablating the studied parameter for that agent, and measuring the resultant producer welfare loss against the correct estimation case. We use the same network parameters as before and assume all users view their feeds rationally. In all ablations, we assume there is no tax. In order to find a producer welfare improving misreport, we perform a linear scan of reported values and take the misreport which produces the largest producer welfare.

First, we evaluate the effect of misestimating $\marginal$, the marginal value of credits in the unbounded time horizon autobidder (Figure \ref{fig:bellman_theta_abl}). We observe that the impact from misestimating $\marginal$ is one sided. Content producers only incur loss when $\marginal$ is overestimated. This is because a higher marginal credit value pushes bids down in each iteration. If $\marginal$ is being misestimated, users may then overcome this by reporting a higher than truthful value of $\pvl{p}$ in order to recover their optimal bid.

\begin{figure}%
    \centering
    { \includegraphics[width=0.98\linewidth]{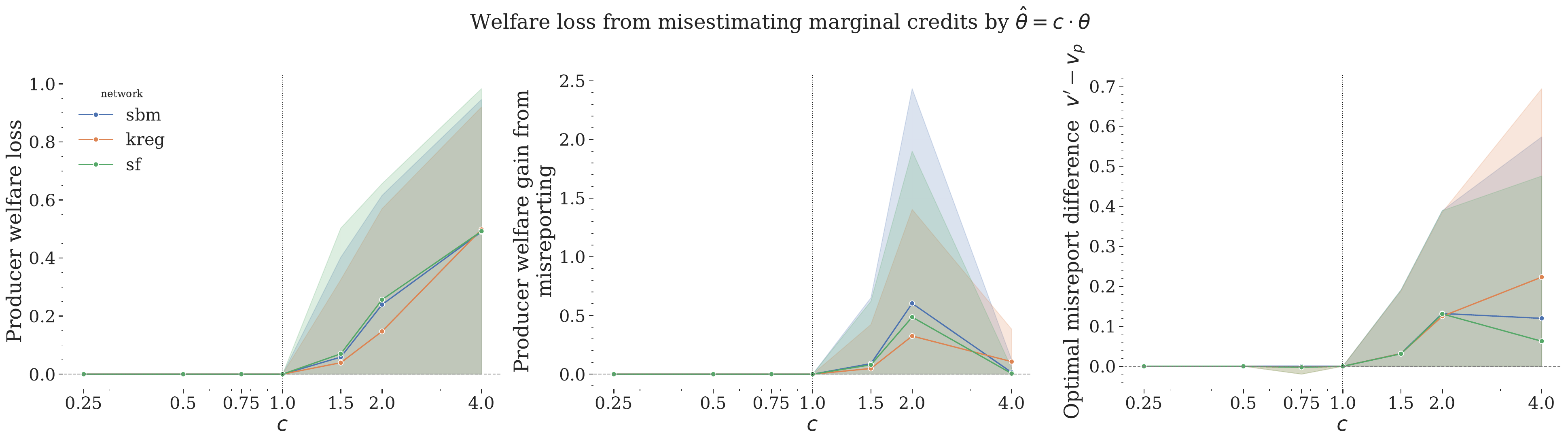} }%
    \caption{The effect of error in estimating $\marginal$ on producer welfare in the unbounded time horizon auction. We report the mean and 95\% percentile bands. The harm is one sided, welfare loss is only incurred when $\marginal$ is estimated to be larger than it is. When $\marginal$ is overestimated, agents are also incentivized to misreport larger private values for their posts.}%
    \label{fig:bellman_theta_abl}%
\end{figure}

Next, we evaluate the effect of misestimating $\nexp$, the expected private value of posts in the myopic autobidder (Figure \ref{fig:myopic_val_abl}). Unlike under the unbounded time horizon, producer welfare loss comes when $\nexp$ is misestimated in either direction. When $\pvl{p} > \nexp$ in the no tax setting, users attempt to expend their entire budget. If $\nexp$ is under reported, reporting the true value of $\pvl{p}$ may cause the autobidder to over spend in this iteration. Conversely, when $\nexp$ is over estimated, bids may become too conservative. Thus, over reporting $\pvl{p}$ recovers the optimal bid.

\begin{figure}%
    \centering
    { \includegraphics[width=0.98\linewidth]{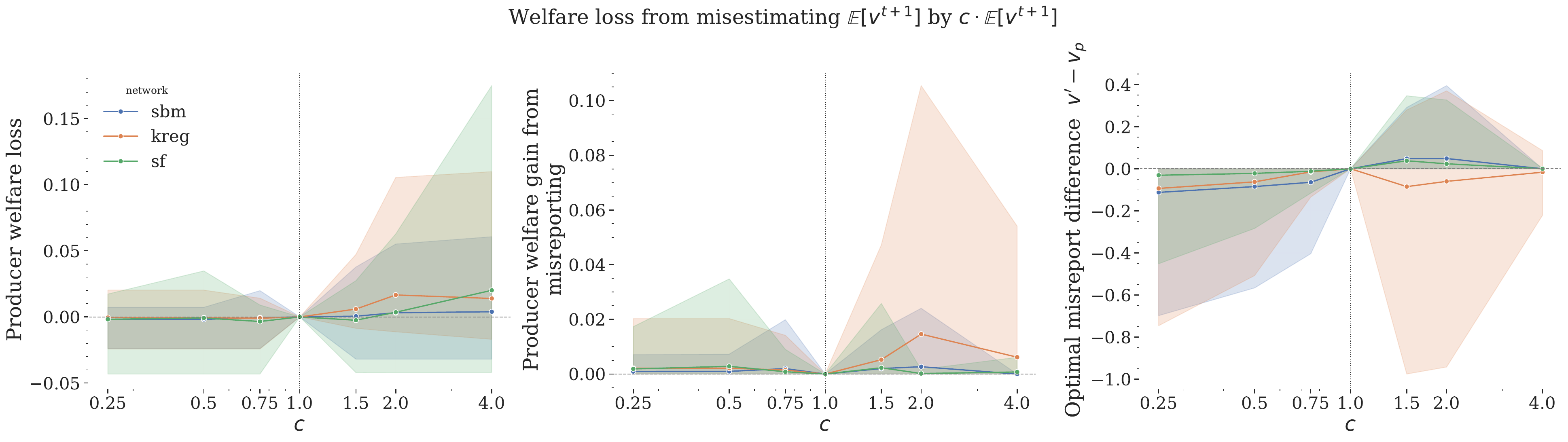} }%
    \caption{The effect of error in estimating $\nexp$ on producer welfare in the myopic auction. We report the mean and 95\% percentile bands. Misestimating $\nexp$ in either direction causes producer welfare to vary substantially. When $\nexp$ is underestimated, users will tend to under report their private values. Conversely, when $\nexp$ is overestimated, users will tend to over report their private values.}%
    \label{fig:myopic_val_abl}%
\end{figure}

We additionally empirically validate our autobidder optimality claim by evaluating the producer welfare change from misreporting under both autobidders (Figure \ref{fig:misreport}). Despite having stronger theoretical guarantees, the myopic autobidder underperforms the Bellman autobidder even when $\nexp$ is correctly estimated. The magnitude of producer welfare gains is small on average and are for the most part due to errors in the estimation of $w$ and $\rho$, the win probability and expected price of a given bid respectively. The Bellman bidder only requires computing one simple numeric solution whereas the myopic autobidder may require two.

\begin{figure}%
    \centering
    \subfloat[Misreporting under the Bellman autobidder]{ \includegraphics[width=0.45\linewidth]{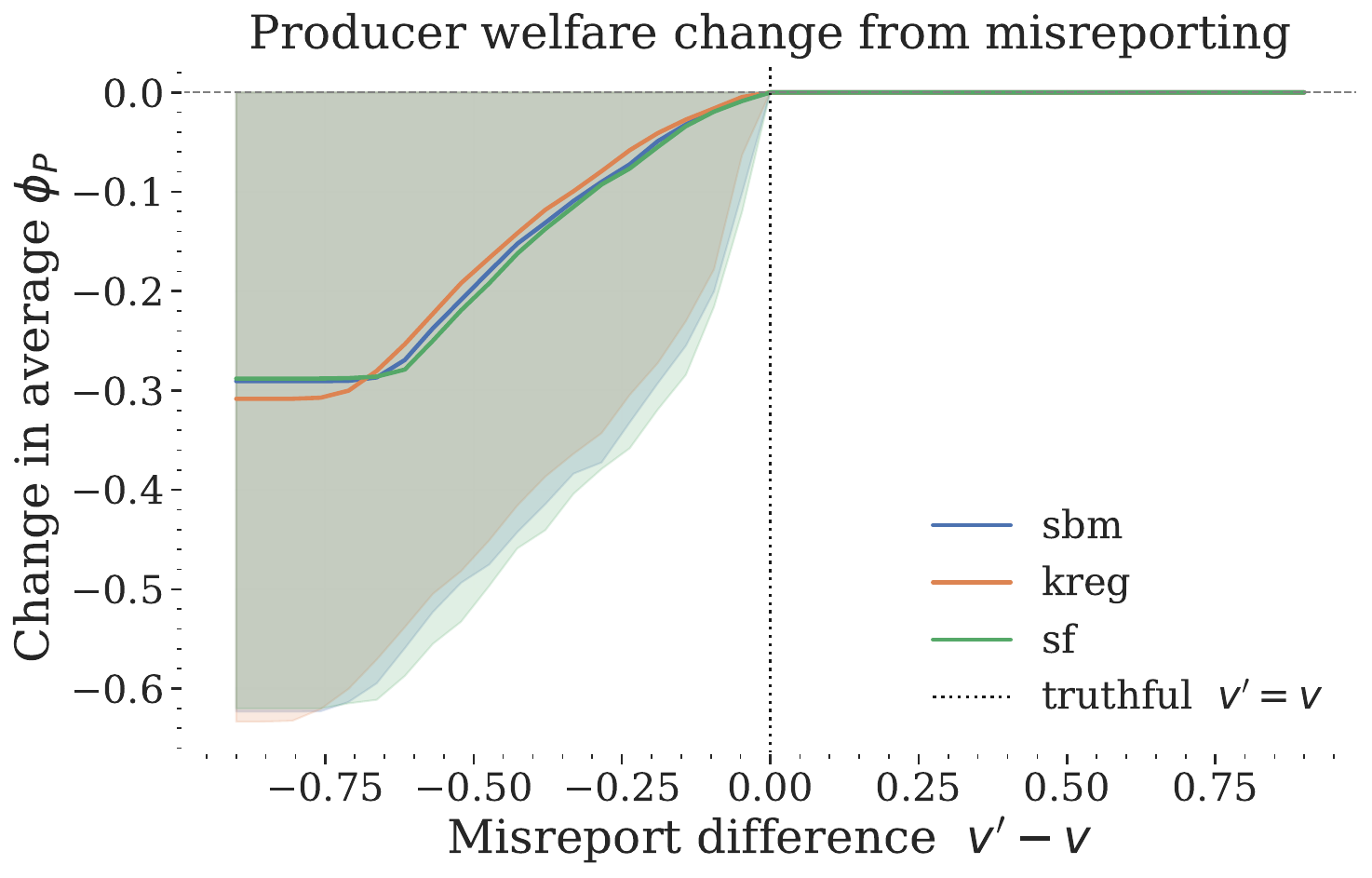} }%
    \subfloat[Misreporting under the myopic autobidder]{ \includegraphics[width=0.45\linewidth]{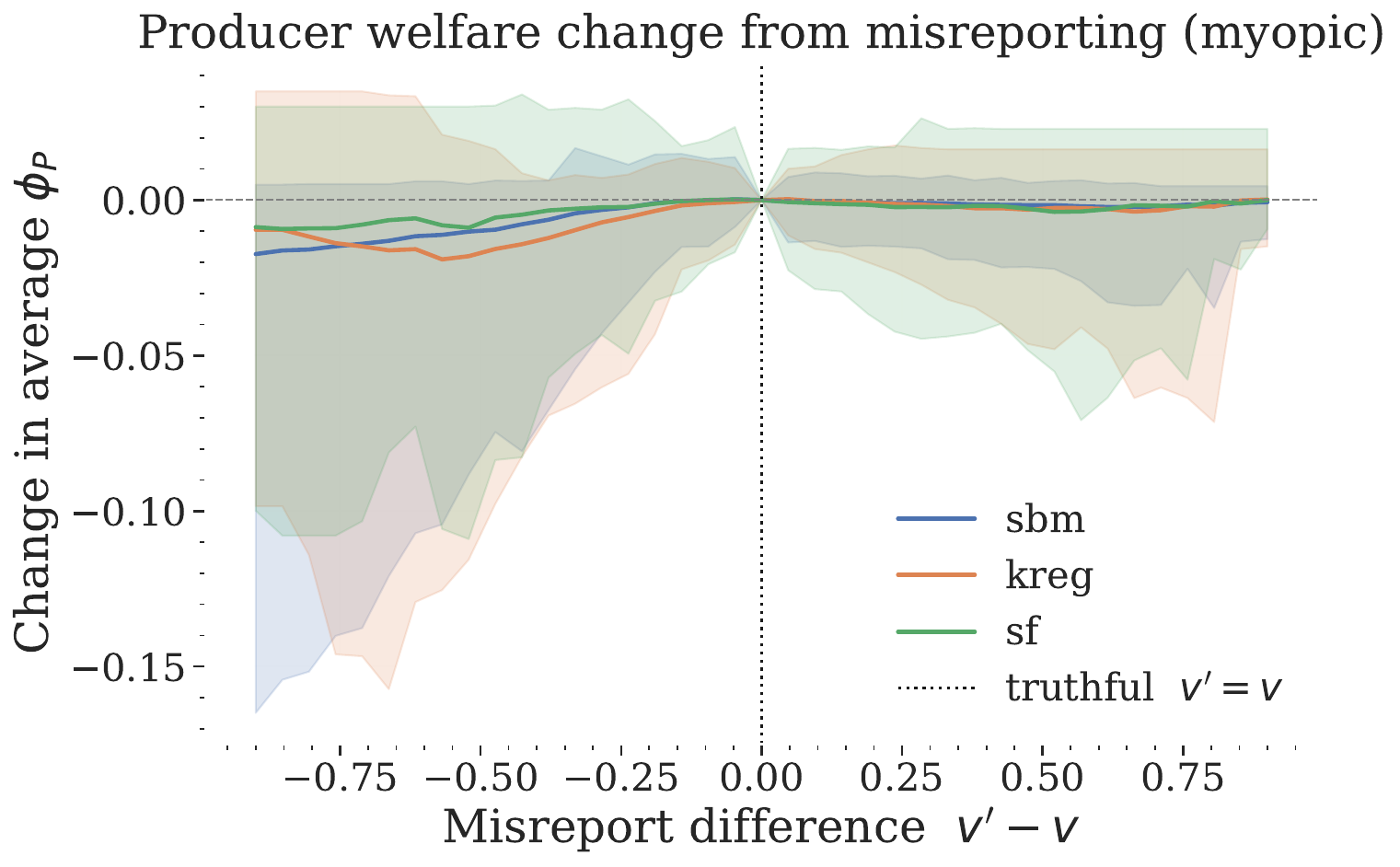} }%
    \caption{Misreporting under the Bellman autobidder is never welfare improving. The myopic autobidder is less robust despite having stronger theoretical guarantees. This is due in large part due to numeric errors from our discrete estimation of $w$ and $\rho$, the probability and price of winning for a given bid respectively. We show the 95\% percentile band to illustrate the range of welfare changes.}%
    \label{fig:misreport}%
\end{figure}

On the negative side, we find that our Bellman autobidder is not robust against coalitional manipulation. In Figure \ref{fig:coalition_ablation} we show that a coalition of users is able to misreport their private values in order to gain producer welfare at the expense of users not in the coalition. We evaluate both randomly chosen coalitions as well as coalitions of low quality content producers. The second case is meant to evaluate the setting in which a group of users is participating in a coordinated misinformation campaign. Our ablation results show that while taxes are able to reduce the gain from misreporting by coalitions, taxes alone are unable to entirely eliminate gain from coalitional misreporting. We also find that the budget capped tax (Section \ref{sec:capped_tax}) is less effective than uncapped power tax (Section \ref{sec:power_tax}).

\begin{figure}%
    \centering
    { \includegraphics[width=0.98\linewidth]{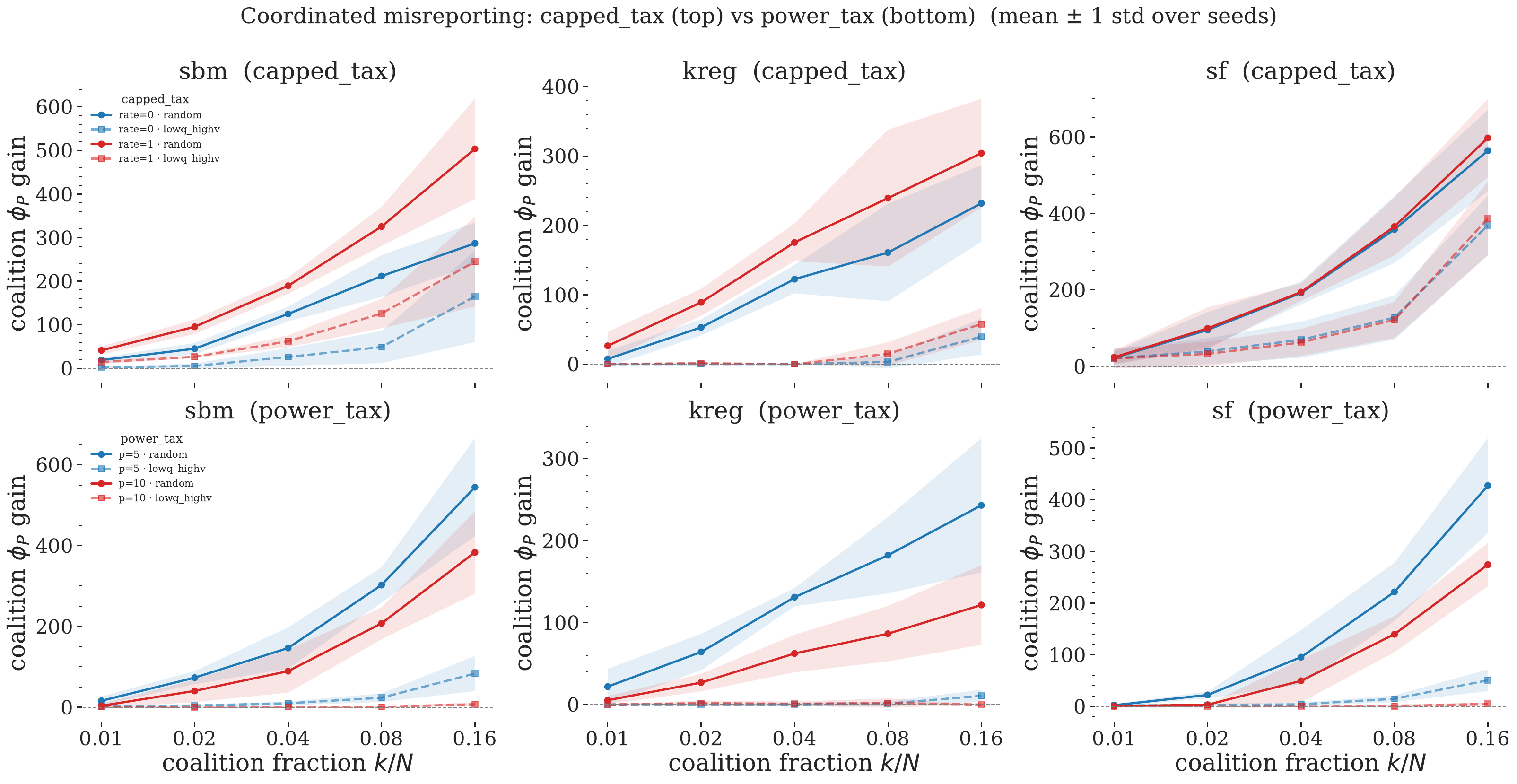} }%
    \caption{A coalition of users is able to systematically vary their reporting in order to increase their producer welfare while reducing the producer welfare of non-coalition users. Evaluations are performed under the Bellman autobidder using both the capped tax as well as the uncapped, power tax. We show that the uncapped tax is able to more aggressively reduce producer welfare gains from the misreporting coalition at the cost of pushing more users out of the effective budget regime.}%
    \label{fig:coalition_ablation}%
\end{figure}

Finally, we evaluate the effect of adding noise to the quality oracle (Figure \ref{fig:quality_ablation}). We evaluate this under the Bellman autobidder with a capped tax rate of $\txr = 1.0$. Under the taxed Bellman feed, we find that producer welfare is moderately lower as noise increases. This is because the noisy budget modification reduces the proportion of agents in the budget regime where private value is highly correlated with bid. Conversely, we find that social welfare modestly increases on average. This is not because of reduced consumption, indeed we fix $\disc = 1.0$ and $\oVwf$ is flat across all noise levels. Rather, the additional noise can push the tax mechanism to further stratify the budgets of high and low quality producers such that high quality producers have larger budgets. This is very noisy however and the effect is very small.

\begin{figure}%
    \centering
    { \includegraphics[width=0.98\linewidth]{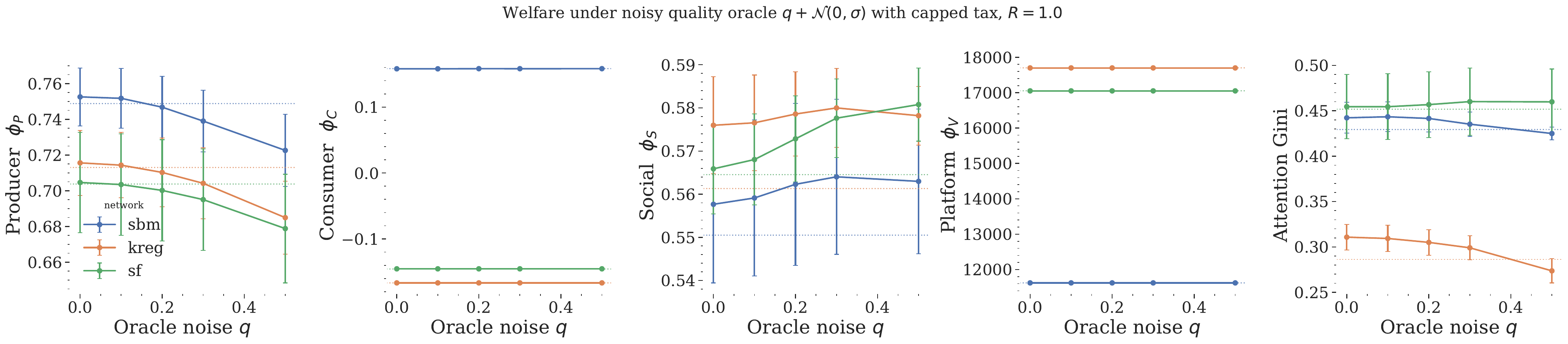} }%
    \caption{We add different levels of Gaussian noise to the quality oracle to determine how different stakeholder welfare measures change. We evaluate this under the Bellman autobidder with a capped tax at $\txr = 1.0$. Adding noise to the quality oracle modestly reduces producer welfare and modestly increases social welfare. On the dotted line we show the welfare with no tax, $\txr=0.0$, which is independent of quality.}%
    \label{fig:quality_ablation}%
\end{figure}

\subsection{Attention distribution}

In addition to consumer and producer welfare, we are concerned with the distribution of attention under each feed. Below, we report the Lorenz curves \cite{lorenz1905methods} and Gini coefficients \cite{gini1997concentration} of impressions received under each feed. We report the results of the myopic autobidder here as the myopic autobidder produces substantially more equitable distributions of attention than the Bellman autobidder. This is despite the relatively poor performance of the myopic autobidder in generating producer welfare. Additionally, we report the results of a \textbf{timeline} based feed which orders posts randomly. We find that the distribution of attention under the Bellman autobidder is relatively invariant under different network topologies. In the $k$-regular autobidder inequality is driven purely by the feed mechanisms as all vertices are topologically identical. When impressions are weighted identically, there is minimal inequality under all feed constructions except the Bellman auction and quality ordered feeds. When impressions are value weighted however, inequality is more apparent though the Bellman auction and quality ordered feeds still exhibit the highest inequality. This no longer holds when follower counts become highly skewed. Under the scale free network which follows a power law distribution, the Bellman auction and quality based feeds exhibit the lowest inequality for both unweighted and value weighted impressions. In the EU core email network however, the two auction based feeds exhibit the lowest inequality in both the weighted and value weighted cases. We report the Gini coefficients as summary statistics for the Lorenz curves below. The myopic auction produces the lowest levels of inequality in all settings except the scale free settings where the Bellman auction is the most equitable.

\begin{figure}%
    \centering
    { \includegraphics[width=0.98\linewidth]{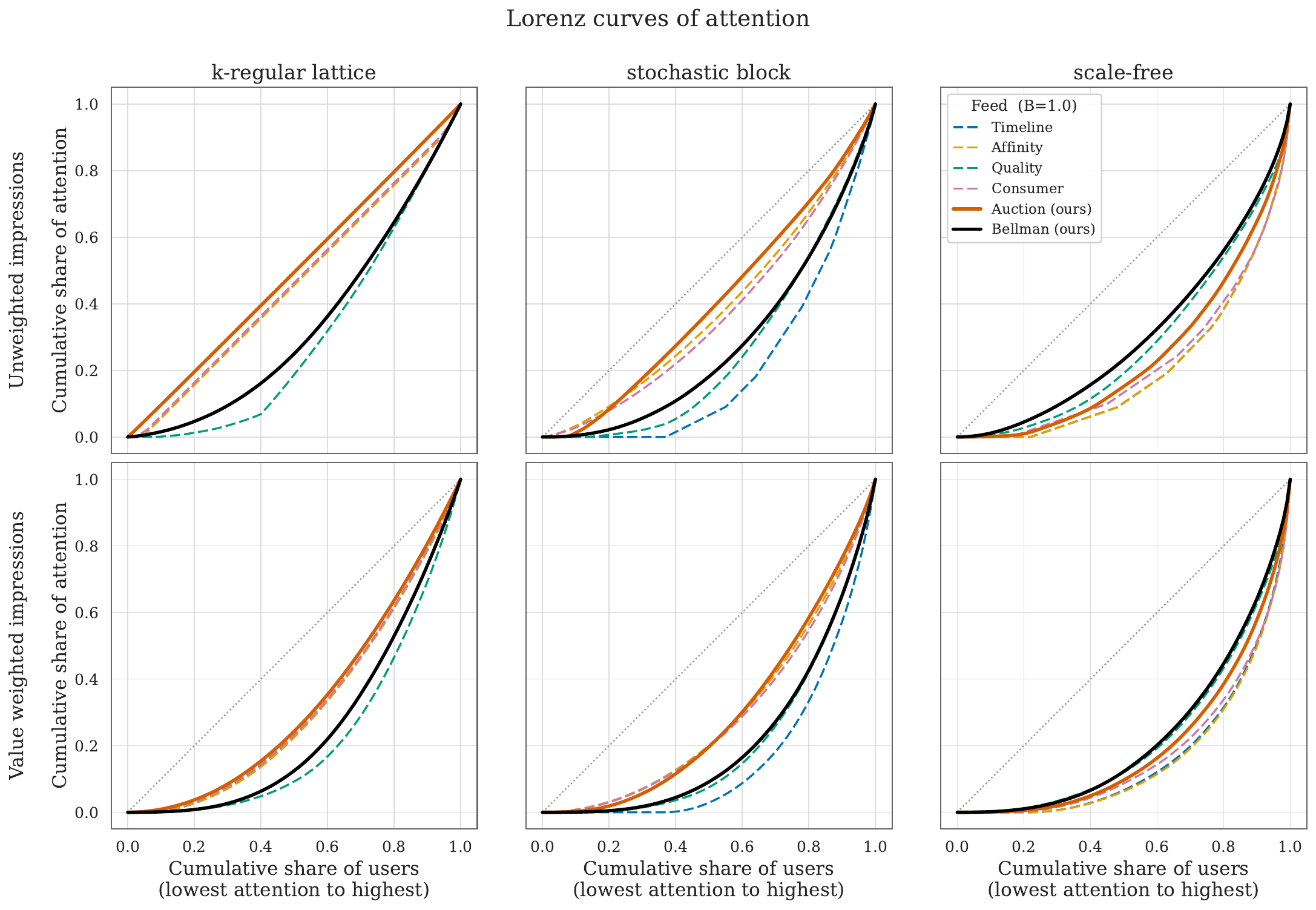} }%
    \caption{The Lorenz curves for inequality of attention across three synthetic networks. Curves further from the diagonal are more unequal. When all vertices are equally positioned, the Bellman auction and quality ordered feeds have the highest inequality. When degree distribution is highly skewed however, the Bellman auction and quality ordered feeds have the lowest inequality. This is because the Bellman auction and quality ordered feeds have attention distributions which are relatively robust to changes in topology whereas other feed algorithms are highly sensitive.}%
    \label{fig:lorenz_curves}%
\end{figure}

\begin{figure}%
    \centering
    { \includegraphics[width=0.98\linewidth]{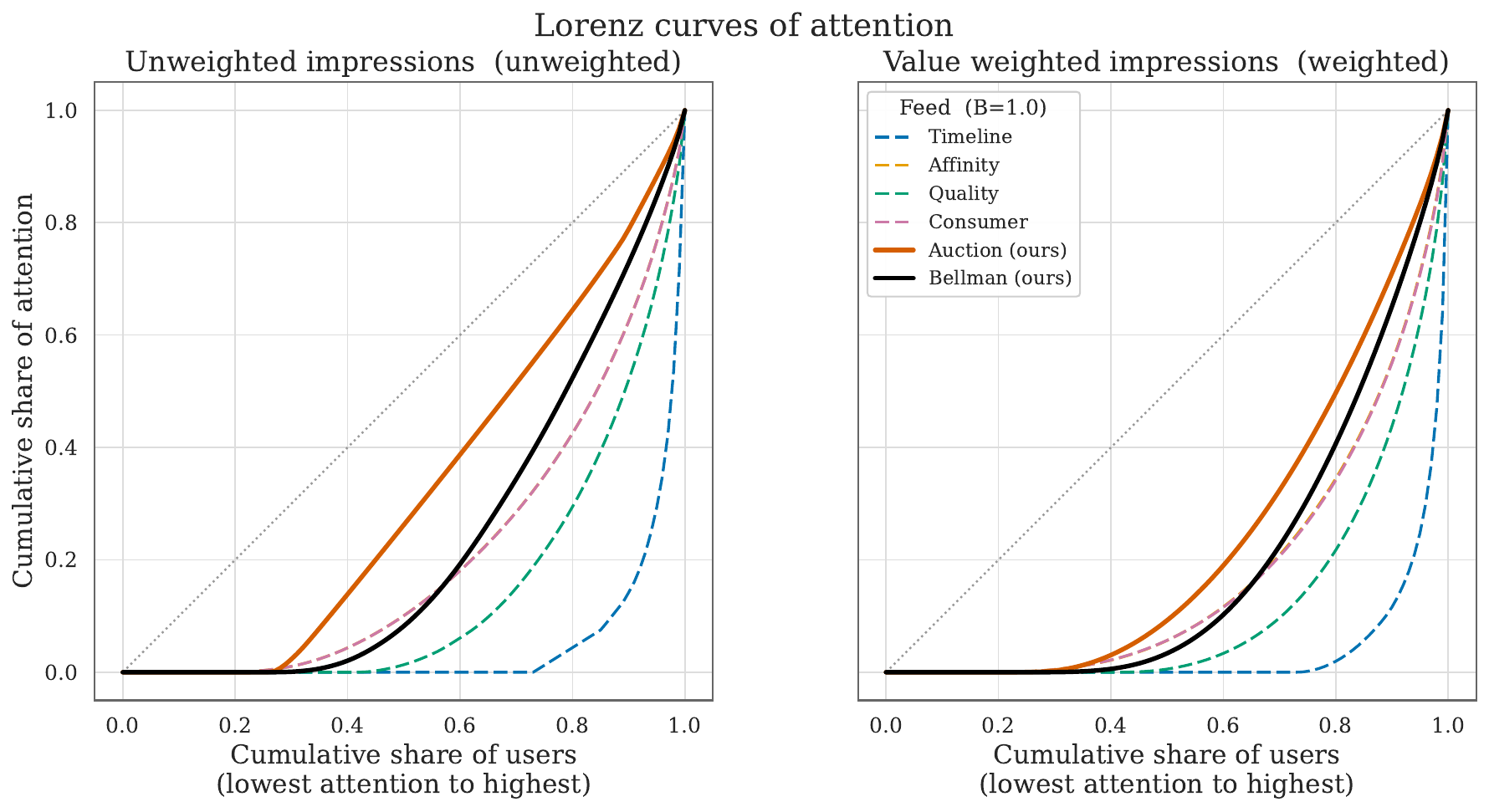} }%
    \caption{The Lorenz curves for inequality of attention on the EU core email network. Curves further from the diagonal are more unequal. Our proposed auction methods present the lowest inequality in attention distribution.}%
    \label{fig:eu_lorenz_curves}%
\end{figure}

\begin{table}
  \centering
  \caption{Gini coefficients for the distribution of attention as measured by impressions for feed algorithms across network types.}
  \label{tbl:gini_coeff}
  \begin{tabular}{l cc cc cc cc}
    \toprule
    & \multicolumn{2}{c}{$k$-regular}
    & \multicolumn{2}{c}{SBM}
    & \multicolumn{2}{c}{Scale free}
    & \multicolumn{2}{c}{EU core e-mail} \\
    \cmidrule(lr){2-3} \cmidrule(lr){4-5} \cmidrule(lr){6-7} \cmidrule(lr){8-9}
    Value weighted? & No & Yes & No & Yes & No & Yes & No & Yes \\
    \midrule
    \multicolumn{9}{@{}l}{\textit{Feed}} \\
    Timeline & 0.083 & 0.370 & 0.591 & 0.669 & 0.598 & 0.669 & 0.907 & 0.923 \\
    Affinity & 0.083 & 0.372 & 0.234 & 0.413 & 0.598 & 0.676 & 0.577 & 0.653 \\
    Quality  & 0.398 & 0.549 & 0.482 & 0.587 & 0.443 & 0.557 & 0.707 & 0.761 \\
    Consumer & 0.070 & 0.362 & 0.267 & 0.426 & 0.567 & 0.645 & 0.577 & 0.655 \\
    Auction (ours) & \textbf{0.008} & \textbf{0.341} & \textbf{0.196} & \textbf{0.411} & 0.518 & 0.603 & \textbf{0.353} & \textbf{0.538} \\
    Bellman (ours) & 0.326 & 0.491 & 0.443 & 0.573 & \textbf{0.393} & \textbf{0.545} & 0.528 & 0.623 \\
    \bottomrule
  \end{tabular}
\end{table}

\section{Limitations and Future Work}\label{sec:limits}

Many aspects of our model are static: the network, post target sets, and content producer strategies are fixed.
In future work, we plan to consider dynamic networks, mutable post target sets, and content producer responses to our tax policy.
Additionally, we plan to consider agents whose attention spans are based on their interest in their feed and whose posting probability is based on attention received.
We also make an assumption that content producers know how much they value their posts. In practice, this is likely untrue and content producers may only be able to provide less granular post values.
While our model allows for weighting users by the likelihood of post sharing, we do not empirically evaluate post sharing. This is omitted for clarity; introducing post sharing complicates bidding. In future work, we plan to model post sharing based on our affinity scores in order to understand whether a taxed auction may be used as an intervention for information spread.
A user level intervention we do not consider is adding a reservation price for feed auctions. This may allow users to filter uninteresting or unimportant, low value posts.

\section{Conclusion}\label{sec:conc}

We introduce, to our knowledge, the first method for allocating attention on social networks via user to user auctions. We equip users with an autobidder which allows them to find optimal bids given the private values they hold for the posts they create. This forms a weakly incentive compatible mechanism wherein users report their true private values to the system.
Complementing the auction, we introduce a tax policy which improves consumer welfare by raising the cost of low consumer welfare posts. We compare our mechanism against baselines via simulation, finding that our methods can deliver on average 36.3\% higher producer welfare than baseline methods in an empirically observed social network.
Additionally, our proposed methods produce a more equitable distribution of attention across a variety of network topologies, including for networks with highly skewed degree distributions.
Finally, we give examples of how RL systems may be used to learn tax policies. Together, these components suggest our mechanism may balance producer and consumer welfare, give a more equal distribution of attention than baselines, align with platform incentives of attention capture and limit the spread of content with negative externalities.

\bibliographystyle{ACM-Reference-Format}
\bibliography{attention}

\appendix
\section{Additional figures}

In Figure \ref{fig:eu_deg} we show the degree distribution for the EU email network (core). This network exhibits the power law degree distribution found in many social networks.
\begin{figure}%
    \begin{center}
        \includegraphics[width=0.5\textwidth]{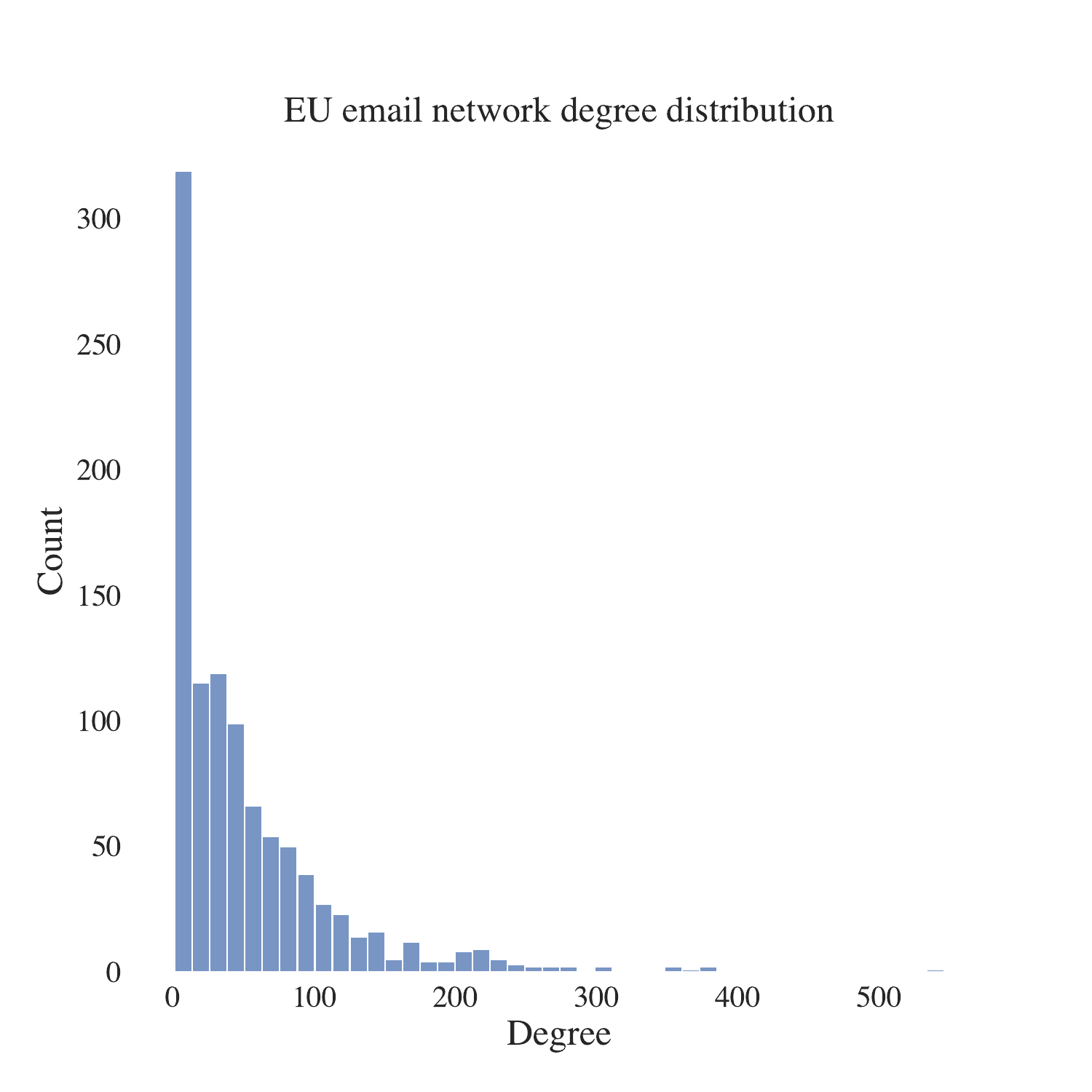}
    \end{center}
    \caption{The degree distribution of the EU email network (core). This exhibits the skewed distribution typical to observed social networks.}%
    \label{fig:eu_deg}%
\end{figure}

\end{document}